\documentclass[twocolumn,preprintnumbers,prd,nofootinbib,superscriptaddress]{revtex4-1}

\usepackage{amsmath}
\usepackage{amssymb}
\usepackage{amsthm}
\usepackage{graphicx}
\usepackage{subfigure}
\usepackage{changepage}
\usepackage{hyperref}
\usepackage{url}
\usepackage{breakurl}
\usepackage{mathtools,xparse}
\usepackage{verbatim}
\usepackage{dsfont}

\theoremstyle{definition}
\newtheorem{theorem}{Theorem}
\newtheorem{corollary}{Corollary}

\newtheorem*{definition*}{Definition}
\newtheorem{lemma}{Lemma}
\newtheorem*{fact}{Fact}

\def\ket#1{| #1 \rangle}

\def\app#1#2{%
  \mathrel{%
    \setbox0=\hbox{$#1\sim$}%
    \setbox2=\hbox{%
      \rlap{\hbox{$#1\propto$}}%
      \lower1.1\ht0\box0%
    }%
    \raise0.25\ht2\box2%
  }%
}

\DeclarePairedDelimiter{\norm}{\lVert}{\rVert}

\begin{document}

\title{Pulling the Boundary into the Bulk}

\author{Yasunori Nomura}
\affiliation{Berkeley Center for Theoretical Physics, Department of Physics, 
  University of California, Berkeley, CA 94720, USA}
\affiliation{Theoretical Physics Group, Lawrence Berkeley National Laboratory, 
  Berkeley, CA 94720, USA}
\affiliation{Kavli Institute for the Physics and Mathematics 
 of the Universe (WPI), The University of Tokyo Institutes for Advanced Study, 
 Kashiwa 277-8583, Japan}
\author{Pratik Rath}
\affiliation{Berkeley Center for Theoretical Physics, Department of Physics, 
  University of California, Berkeley, CA 94720, USA}
\affiliation{Theoretical Physics Group, Lawrence Berkeley National Laboratory, 
  Berkeley, CA 94720, USA}
\author{Nico Salzetta}
\affiliation{Berkeley Center for Theoretical Physics, Department of Physics, 
  University of California, Berkeley, CA 94720, USA}
\affiliation{Theoretical Physics Group, Lawrence Berkeley National Laboratory, 
  Berkeley, CA 94720, USA}

\begin{abstract}
Motivated by the ability to consistently apply the Ryu-Takayanagi 
prescription for general convex surfaces and the relationship between 
entanglement and geometry in tensor networks, we introduce a novel, 
covariant bulk object---the holographic slice.  The holographic slice 
is found by considering the continual removal of short range information 
in a boundary state.  It thus provides a natural interpretation as the 
bulk dual of a series of coarse-grained holographic states.  The slice 
possesses many desirable properties that provide consistency checks for 
its boundary interpretation.  These include monotonicity of both area 
and entanglement entropy, uniqueness, and the inability to probe beyond 
late-time black hole horizons.  Additionally, the holographic slice 
illuminates physics behind entanglement shadows, as minimal area extremal 
surfaces anchored to a coarse-grained boundary may probe entanglement 
shadows.  This lets the slice flow through shadows.  To aid in developing 
intuition for these slices, many explicit examples of holographic slices 
are investigated.  Finally, the relationship to tensor networks and 
renormalization (particularly in AdS/CFT) is discussed.
\end{abstract}

\maketitle

\makeatletter
\def\l@subsection#1#2{}
\def\l@subsubsection#1#2{}
\makeatother
\tableofcontents

\section{Introduction}
\label{sec:intro}

The holographic duality between asymptotically \mbox{Anti-de~Sitter} (AdS) 
spacetimes in $d+1$ dimensions and conformal field theories (CFTs) in 
$d$ dimensions is perhaps the closest realization of a complete theory 
of quantum gravity~\cite{Maldacena:1997re,Gubser:1998bc,Witten:1998qj}. 
One of the most intriguing results stemming from this correspondence 
is the renowned Ryu-Takayanagi formula relating entanglement entropy 
in time independent CFTs to the area of minimal surfaces in the bulk 
spacetime~\cite{Ryu:2006bv}.  The covariant extension of this formula 
to include time dependent cases was obtained by Hubeny-Rangamani-Takayanagi 
and uses extremal bulk surfaces (henceforth referred to as the 
HRRT prescription)~\cite{Hubeny:2007xt}.  Remarkably, by including 
quantum corrections to this formula, one obtains entanglement wedge 
reconstruction~\cite{Jafferis:2015del,Dong:2016eik}.  These investigations 
have shed light on the deep connection between entanglement in the 
boundary and emergent gravitational physics in the bulk.

However, there is reason to believe that these results extend beyond the 
scope of AdS/CFT.  The Bekenstein-Hawking formula~\cite{Bekenstein:1973ur,%
Hawking:1974sw} and the covariant entropy bound~\cite{Bousso:1999xy} 
provide us with holographic bounds on entropy in general spacetimes. 
These suggest that gravitational physics may inherently be 
holographic~\cite{'tHooft:1993gx,Susskind:1994vu}.  Furthermore, 
the areas of extremal surfaces anchored to any convex boundaries 
satisfy all known entropic inequalities~\cite{Miyaji:2015yva,Bao:2015bfa,%
Sanches:2016sxy,Rota:2017ubr}.  In isometric tensor networks, calculating 
the entanglement entropy of a subregion of boundary sites reduces 
to finding the minimum cut across the network~\cite{Swingle:2012wq}. 
A version of entanglement wedge reconstruction also holds in perfect 
and random tensor networks~\cite{Pastawski:2015qua,Hayden:2016cfa}. 
This evidence seems to indicate that the HRRT prescription may in fact 
generalize to spacetimes outside of AdS.

It is with this perspective that we have pursued investigations of 
quantum gravity beyond AdS/CFT.  We postulate that the HRRT prescription 
(with quantum corrections~\cite{Faulkner:2013ana,Engelhardt:2014gca} to 
allow for entanglement wedge reconstruction) applies to general convex 
boundaries.  In particular, we assume the existence of a quantum state 
that ``lives'' on the convex boundary and encodes the bulk geometry of 
the interior.  Our previous work~\cite{Nomura:2016ikr,Nomura:2017npr,%
Nomura:2017fyh} has focused primarily on investigating this assumption 
applied to holographic screens~\cite{Bousso:1999cb}, but holographic 
screens are only special in the sense that they are the largest surfaces 
in which we can apply the HRRT prescription in general spacetimes.  In 
this paper, we have relaxed this condition and look at general convex 
surfaces---in particular this allows us to consider asymptotically AdS 
spacetimes.

We emphasize that the postulate we adopt here is falsifiable.  At any 
point in the analysis, had the geometric properties of general relativity 
prohibited a consistent boundary interpretation, the program would 
have failed.  Through the present, however, no such roadblock has 
presented itself.  In fact, we can view the self-consistency of this 
work as further evidence that the relationship between entanglement 
and geometry prevails in general spacetimes.

The present framework allows us to consider a nested family of convex 
surfaces each of which contains less bulk information than the previous. 
Taking a natural continuum version of these convex surfaces yields 
a surface which we dub the {\it holographic slice}.  As has proved 
historically useful, studying covariantly defined geometric objects 
yields insights into holographic theories, and the holographic slice 
is such an object.  In particular, the slice allows us to visualize 
the coarse-graining of holographic states and provides us with a method 
to categorize bulk regions as being ``more IR'' than others.

The construction of the slice is purely perturbative, and thus does 
not allow us to analyze complex quantum gravitational states formed 
as superpositions of many geometries.  It is inherently tied to one 
background geometry, but this is no more restrictive than the HRRT 
prescription itself.  For a simply connected boundary, the slice seems 
to sweep the maximal bulk region that can be perturbatively reconstructed, 
i.e.\ it pulls the boundary in through shadows all the way to black 
hole horizons, and contracts no further.  It was the study of extremal 
surfaces in maximally entangled pure states that initiated this work, 
where it was realized that maximally entangled states were fixed points 
of flow in the direction of small HRRT surfaces~\cite{Nomura:2017fyh}.%
\footnote{By maximally entangled, we mean that the von~Neumann entropy 
 of any subregion saturates the Page curve~\cite{Page:1993df}.  This 
 was referred to as maximally entropic in Ref.~\cite{Nomura:2017fyh}.}

\subsection*{Overview}

In Section~\ref{sec:motivation}, we offer the intuitive motivation for 
the construction of the holographic slice.  This reveals the connection 
to coarse-graining and provides a basic definition for the object. 
In Section~\ref{sec:holo_slice}, we geometrically define the 
object rigorously and then illustrate some of its most important 
properties.  These properties must necessarily be satisfied for 
a consistent interpretation as coarse-graining of a boundary state. 
Section~\ref{sec:example} goes through multiple explicit examples of 
the slice in different spacetimes.

After introducing the holographic slice as a geometric object, we are poised 
to analyze its boundary interpretation in Section~\ref{sec:interpret}. 
Here we delve into its relationship to coarse-graining and emphasize that 
the slice encodes a sequence of \mbox{codimension-0} bulk regions, not 
merely the \mbox{codimension-2} bulk convex surfaces.  We also describe 
the relationship to tensor networks, particularly continuous tensor 
networks~\cite{Haegeman:2011uy}.  Because the holographic slice is 
constructed from one boundary time slice, it can be used as a novel 
gauge fixing of the bulk---this is discussed in the final subsection 
of Section~\ref{sec:interpret}.  Since the slice grants us a way to 
uniquely pull in the boundary, it is natural to consider its connection 
to renormalization.  This is explored in Section~\ref{sec:renorm}.  We 
conclude with discussing the slice's place in the wider view of quantum 
gravity in Section~\ref{sec:disc}.  The appendices contain proofs of 
various geometric statements made in the body of the text.

\subsection*{Preliminaries}

This paper will work in the framework of holography for general 
spacetimes proposed in Ref.~\cite{Nomura:2016ikr}.  We will highlight 
the applications to AdS/CFT but use language from generalized holography. 
In particular, the term ``boundary'' will refer to the holographic 
screen, which reduces to the conformal boundary of AdS.  Additionally, 
holographic screens have a unique time foliation into \mbox{codimension-2} 
surfaces called leaves~\cite{Bousso:2015mqa}.  This uniqueness is lost 
in the case of AdS because of the asymptotic symmetries, but to remain 
consistent within the generalized framework we must choose a particular 
time foliation of the boundary of AdS~\cite{Nomura:2017npr}.  We will 
then refer to a time slice of this foliation as a leaf.

For a subregion, $A$, of a leaf, we will denote its HRRT surface as 
$\gamma(A)$ and its entanglement wedge as $\text{EW}(A)$.  $\text{EW}(A)$ 
is defined as the domain of dependence of any closed, compact, achronal 
set with boundary $A \cup \gamma(A)$.%
\footnote{In standard AdS/CFT, reflective boundary conditions are imposed 
 at the conformal boundary.  This extends the domain of dependence of 
 $A \cup \gamma(A)$ to include the boundary domain of dependence of $A$. 
 However, for holographic screens there are no such impositions on the 
 screen (they can even be spacelike), and thus $\text{EW}(A)$ generally 
 does not include any portion of the screen other than $A$ itself. 
 In particular, there is no generalization of causal wedges to 
 holographic screens.}
Throughout this paper, we will work at the lowest order in bulk Newton's 
constant.  In particular, we will only consider extremal surfaces 
found by extremizing the area, not the generalized entropy.  We 
expect that by making appropriate modifications, along the lines of 
Refs.~\cite{Faulkner:2013ana,Engelhardt:2014gca,Bousso:2015eda}, our 
results can be extended to higher orders.

For an achronal \mbox{codimension-2} surface $\omega$, we denote the 
domain of dependence of any achronal set with boundary $\omega$ as 
$D(\omega)$.  The definition of convexity for a \mbox{codimension-2} 
surface is given in Appendix~\ref{app:convexity}.
Wherever $G_{\rm N} = l_{\rm P}^{d-1}$ appears, it represents 
Newton's constant in the bulk spacetime.

\section{Motivation}
\label{sec:motivation}

Amongst other things, the concept of subregion duality in holographic theories 
allows us to address questions regarding where bulk information is stored in 
the boundary theory~\cite{Hamilton:2006az,Hamilton:2006fh,Heemskerk:2012mn,%
Bousso:2012sj,Czech:2012bh,Wall:2012uf,Headrick:2014cta,Jafferis:2015del,%
Dong:2016eik}.  This line of inquiry has provided us with the intuition 
that bulk geometric information is encoded redundantly in the boundary 
theory.  In particular, a bulk local operator can be represented in 
multiple different regions of the boundary, a special case being the 
whole boundary space.  However, despite this seemingly democratic 
distribution of bulk information throughout boundary degrees of freedom, 
lack of access to a boundary region necessarily prohibits the reconstruction 
of a corresponding bulk region.  Namely, if one removes a subregion $A$ 
from a leaf $\sigma$, the maximum possible bulk region reconstructed from 
the remaining region, $\overline{A}$, will be the entanglement wedge of 
$\overline{A}$, $\text{EW}(\overline{A})$.  This implies that indispensable 
information of the region $\text{EW}(A)$ is stored in $A$.

Suppose one were to coarse-grain over all boundary subregions of balls 
of radius $\delta$.  From the logic above, the bulk region whose information 
can remain is given by
\begin{equation}
  R(B_\delta) = \bigcap\limits_{p \in \sigma} 
    \text{EW}\bigl(\overline{B_\delta(p)}\bigr),
\label{eq:bulk-remain}
\end{equation}
where $B_\delta(p)$ is a ball of radius $\delta$ centered at point $p$ 
on the boundary leaf $\sigma$.  Because the intersection of domains of 
dependence is itself a domain of dependence,%
\footnote{We could not find a proof of this statement, so we have 
 included one in Appendix~\ref{app:domains}.}%
$R(B_\delta)$ is a domain of dependence of some achronal set, and thus 
has a unique boundary, $\sigma(B_\delta)$.

Motivated by ideas of holographic renormalization in 
AdS/CFT~\cite{Akhmedov:1998vf,Alvarez:1998wr,Balasubramanian:1999jd,%
Skenderis:1999mm,deBoer:1999tgo}, surface state 
correspondence~\cite{Miyaji:2015yva}, and previous work on 
holographic screens~\cite{Sanches:2016sxy,Nomura:2016ikr,Nomura:2017npr,%
Nomura:2017fyh}, we conjecture that there exists a holographic 
state living on $\sigma(B_\delta)$ which encodes aspects of the bulk 
geometry to its interior.  A check of this proposal is that the HRRT 
prescription can be consistently applied, in the sense that the areas 
of these HRRT surfaces obey the known holographic entropy inequalities. 
Given this consistency check, we conjecture that ``coarse-grained 
subregion duality'' holds---namely that entanglement wedge reconstruction 
holds on this new leaf.

Now suppose one wants to coarse-grain over some scale on this new leaf, 
$\sigma(B_\delta)$.  This will produce a new leaf even deeper in the 
bulk. Repeating this process will in turn produce a series of new leaves, 
henceforth called renormalized leaves.  Sending the coarse-graining scale 
at each step to zero in a consistent manner will produce a continuous 
family of renormalized leaves that sweep out a smooth surface through 
the bulk, $\Upsilon$.  The manner in which $\Upsilon$ is constructed 
naturally reveals its relationship to holographic coarse-graining. 
This prompts us to assert that the continuous coarse-graining of a 
holographic state pulls the boundary slice in along the slice $\Upsilon$.

\section{Holographic Slice}
\label{sec:holo_slice}

The geometric object $\Upsilon$ is what we will refer to as the holographic 
slice.  In this section, we give a more rigorous definition of holographic 
slices and highlight some of the salient properties of them.

\subsection{Definition}
\label{subsec:def}

Consider a closed \mbox{codimension-2} achronal surface $\sigma$ living 
in a $(d+1)$-dimensional spacetime $M$.  Denote the two future-directed 
null orthogonal directions as $k$ and $l$.  Suppose the null expansions 
along these directions satisfy $\theta_k \le 0$ and $\theta_l > 0$.%
\footnote{Appropriate modifications can be made to extend to surfaces 
 where $\theta_k \ge 0$ and $\theta_l < 0$.}
For concreteness, one could imagine $\sigma$ to be a leaf of a past 
holographic screen ($\theta_k = 0$, $\theta_l > 0$) or a time slice of 
the (regularized) boundary of AdS.  Borrowing this language, we will 
call $\sigma$ a leaf.  From Ref.~\cite{Sanches:2016sxy}, we know that 
the boundary of the domain of dependence of $\sigma$, $D(\sigma)$, is 
an extremal surface barrier for HRRT surfaces anchored to $\sigma$. 
In addition, the boundary of an entanglement wedge of a subregion 
$\Gamma$ on $\sigma$ serves as an extremal surface barrier for all 
extremal surfaces anchored within $\text{EW}(\Gamma)$.

Now, on $\sigma$, consider a family of open, \mbox{codimension-0} (within 
the leaf) smooth subregions, with an injective mapping from points on the 
leaf, $p$, to subregions, $C(p)$, with the constraint that $p \in C(p)$ 
and that $C(p)$ varies continuously with $p$.  For example, one may take 
$C(p)$ to be open balls of radius $\delta$ centered at $p$, $B_\delta(p)$. 
Now, let
\begin{equation}
  R(C) = \bigcap\limits_{p \in \sigma} \text{EW}\bigl(\overline{C(p)}\bigr),
\label{eq:domain1}
\end{equation}
where $\overline{C(p)}$ is the complement of $C(p)$ in $\sigma$.  From 
Appendix~\ref{app:domains}, we know that $R(C)$ itself is a domain of 
dependence of some achronal sets, all of which share a unique boundary, 
$\sigma_C^1$, called a renormalized leaf; see Fig.~\ref{fig:RG}.
\begin{figure}[t]
  \includegraphics[scale=0.25]{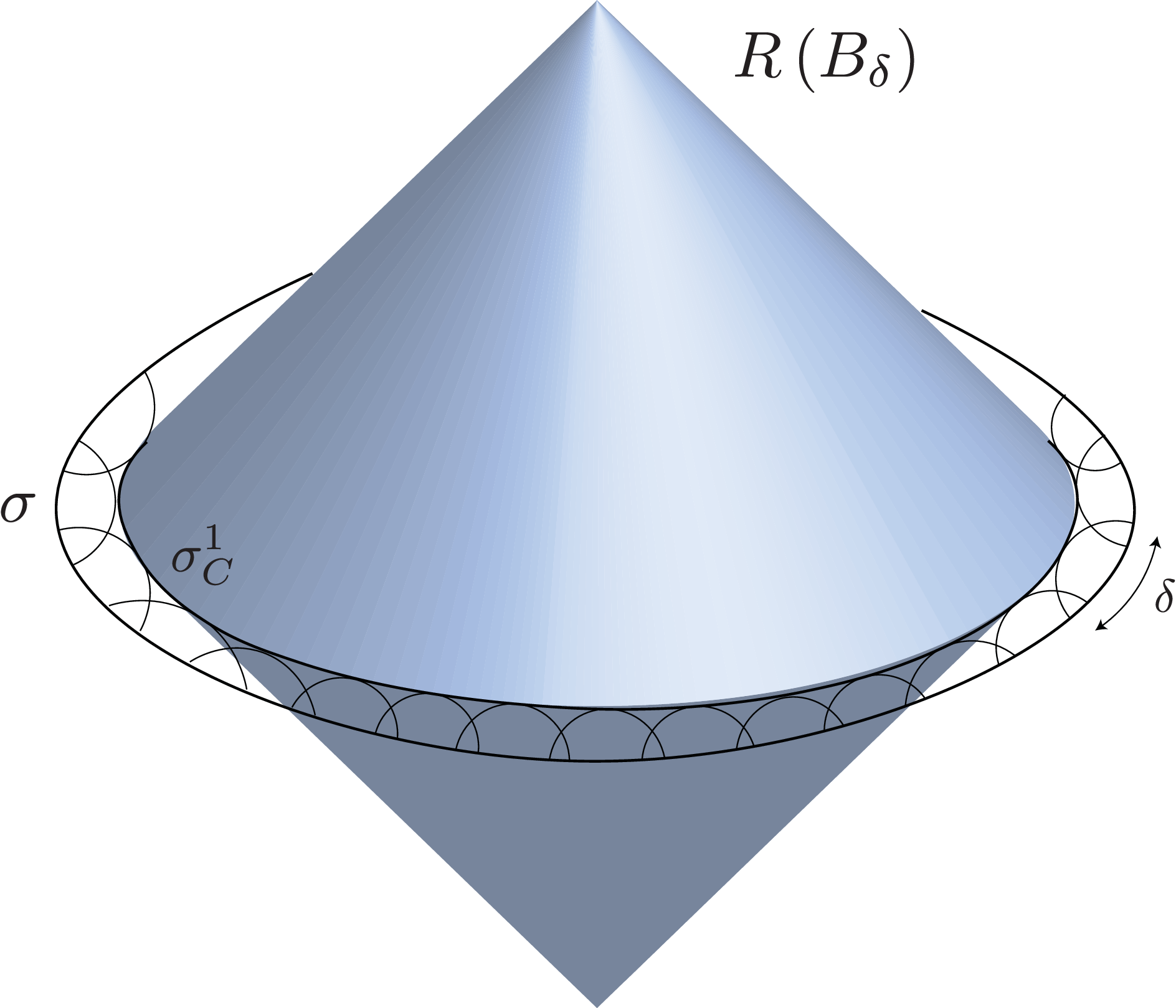} 
\caption{$R(B_{\delta})$ is the entanglement wedge associated with the 
 new leaf $\sigma_C^1$, where we have taken $C(p) = B_\delta(p)$.  It 
 is formed by intersecting the entanglement wedges associated with the 
 complements of spherical subregions of size $\delta$ on the original 
 leaf $\sigma$.}
\label{fig:RG}
\end{figure}

Provided the characteristic scale of $\text{EW}(C(p))$ is sufficiently 
smaller than the extrinsic curvature scale of $\sigma$, we can identify 
all points on $\sigma_C^1$ with those on $\sigma$.  At each point on 
$\sigma$, consider the plane generated by $k$ and $l$.  This will intersect 
$\sigma_C^1$ at one point so long as adjacent planes do not intersect at a 
caustic before hitting $\sigma_C^1$, which is guaranteed if we take $C(p)$ 
to be sufficiently small.  We now have a natural identification of the points 
on $\sigma$ with points on $\sigma_C^1$.  In particular, we can identify the 
length scale $\delta$ on $\sigma$ to the appropriate scale on $\sigma_C^1$. 
Moreover, since $\sigma_C^1 \subseteq \text{EW}(\overline{C(p)})\,\, 
\forall p$ and $\text{EW}(\overline{C(p)})$ serves as an extremal surface 
barrier for all surfaces anchored within $\text{EW}(\overline{C(p)})$, 
all extremal surfaces anchored to $\sigma_C^1$ are contained within 
$R(C) = D(\sigma_C^1)$.  If we interpret the area of these extremal 
surfaces as entanglement entropies for the associated subregions, convexity 
of $R(C)$ and the null energy condition ensure that all known holographic 
entanglement inequalities are satisfied.  This allows us to interpret 
these extremal surfaces as HRRT surfaces.

Utilizing coarse-grained subregion duality, we can repeat the construction 
above but with new subregions, $C^1(p)$, on the renormalized leaf, 
$\sigma_C^1$.  This yields a new leaf $\sigma_C^2$.  We can associate 
points on $\sigma_C^2$ with those on $\sigma_C^1$, and hence on $\sigma$, 
as was done above.  All we require of the new subregions is that the 
size scale of $C^i(p)$ match with that of $C(p)$ under the natural 
identification described previously.  This procedure can be repeated 
until stringy effects become important, i.e.\ when the size of the 
renormalized leaf is $O(l_{\rm s})$.

In the limit of sending the size of $C(p)$ to $0$ for all $p$, the 
collection of all renormalized leaves in $M$, $\Upsilon = \{ \sigma_{C}^i \}$ 
sweeps out a continuous surface.  This is a holographic slice.  Note that 
we can take the $G_{\rm N} \rightarrow 0$ limit in discussing classical 
spacetime, so we may take $C(p) \rightarrow 0$ thus making $\Upsilon$ 
continuous in this limit.  We can then label the renormalized leaves of 
$\Upsilon$ by some continuous parameter $\lambda$, corresponding to the 
depth of the renormalized leaf, i.e.\ $\sigma(\lambda)$ is some $\sigma_C^i$ 
and $\sigma(0) = \sigma$.  Below we take $\lambda$ to decrease as $i$ 
increases, so that $\lambda \leq 0$.

\subsection{Properties}
\label{subsec:prop}

\subsubsection{Uniqueness}
\label{subsubsec:unique}

$\Upsilon$ is dependent on an extraordinary number of degrees of freedom, 
namely the shape $C(p)$ at each point on $\sigma$.  Despite this freedom, 
we find that $\Upsilon$ is unique provided some mild assumptions hold.

In particular, imposing homogeneity and isotropy on $C(p)$ (and each 
subsequent $C^i(p)$) yields a unique holographic slice.  For example, 
one can restrict themselves to the case where $C^i(p)$ are composed of 
the same shape and with random orientations.  These all reduce to the 
slice formed by considering balls of constant radius for $C(p)$ and 
mapping these balls to the subsequent renormalized leaves.  We will 
focus on this preferred slice for the remainder of the paper.

A full discussion of uniqueness is provided in Appendix~\ref{app:uniqueness}. 
However, one of the primary results is that the vector $s$, which is 
tangent to $\Upsilon$ and radially evolves the leaf inward is given by
\begin{equation}
  s = \frac{1}{2} (\theta_k l + \theta_l k).
\label{eq:s-vec}
\end{equation}
This tells us that for a (non-renormalized) leaf of a holographic screen, 
$s \propto k$ and $\theta_s = 0$.  In these situations, the holographic 
slice initially extends in the null direction and the leaf area remains 
constant.

In fact, the $s$ vector coincides with the Lorentzian generalization of 
the mean curvature vector of the leaf~\cite{Andersson:2007gy,Mars:2014wbd}. 
This preferred holographic slice is then realized as the mean curvature 
flow of the initial leaf.

\subsubsection{Monotonicity of Renormalized Leaf Area}

One of the most important features of the holographic slice is that the area 
of the renormalized leaves decreases monotonically as $\lambda$ decreases. 
This is crucial to the interpretation as coarse-graining.  This property can 
be shown in a manner similar to Ref.~\cite{Sanches:2016pga}, but only after 
showing that $\theta_k \leq 0$ and $\theta_l \geq 0$ for each renormalized 
leaf.  This is proved in Appendix~\ref{app:convexity}.

Armed with this knowledge, consider a point $p$ on $\sigma$ and the $s$ 
vector (defined in the previous section) orthogonal to the leaf $\sigma$ 
at $p$.  The integral curves of $s$ passing through $p$ provide a mapping 
of $p$ to a unique point on each $\sigma(\lambda)$.  Now consider an 
infinitesimal area element $\delta A$ around $p$.  The rate at which 
this area changes as one flows along $s$ is measured by the expansion
\begin{equation}
  \theta_s = \theta_k \theta_l \leq 0,
\label{eq:theta_s}
\end{equation}
as found in Appendix~\ref{app:uniqueness}.

Since the area for all infinitesimal area elements decreases on flowing 
along $s$, the total leaf area also decreases.  In fact, this property holds 
locally.  For any subregion of a renormalized leaf, as one flows inward 
along the holographic slice the area of the subregion decreases monotonically.

\subsubsection{Monotonicity of Entanglement Entropy}

Along with the fact that the renormalized leaf area shrinks, the entanglement 
entropy of subregions also decreases monotonically.  This must necessarily 
happen for a consistent interpretation that the coarse-graining procedure 
continuously removes short range entanglement.  This is precisely the 
spacelike monotonicity theorem of Ref.~\cite{Nomura:2016ikr}, so we will 
only sketch the idea.%
\footnote{This interpretation may in fact be the most natural explanation 
 of why the spacelike monotonicity theorem holds.}

Consider a leaf $\sigma = \sigma(\lambda_0)$ and a renormalized 
leaf obtained after some small amount of radial evolution, $\sigma' 
= \sigma(\lambda_0 + \delta\lambda)$, where $\delta\lambda < 0$. 
A subregion $A$ of $\sigma$ is mapped to subregion $A'$ of $\sigma'$ 
by following the integral curves of $s$.

Suppose the HRRT surface $\gamma(A')$ anchored to $A'$ is the minimal 
surface on a spacelike slice $\Sigma'$.  Now extend $\Sigma'$ by including 
the portion of the holographic slice between $\sigma$ and $\sigma'$, such 
that $\Sigma'$ is now an achronal slice with boundary $\sigma$.  Now consider 
the minimal surface $\Xi(A)$ anchored to $A$ on this extended $\Sigma'$. 
$\Xi(A)$ has a portion in the exterior of $\sigma'$ which can be projected 
down to $\sigma'$ using the normal vector $s$.  This projection, denoted 
by $\Xi(A) \rightarrow \pi(\Xi(A))$, decreases the area of $\Xi(A)$ due to 
the spacelike signature of $\Sigma'$.  This projection results in a surface 
anchored to $A'$, which must have an area greater than that of $\gamma(A')$. 
On the other hand, due to the maximin procedure~\cite{Wall:2012uf}, 
$\Xi(A)$ must also have an area less than the area of the HRRT surface 
$\gamma(A)$ anchored to $A$.  In summary,
\begin{equation}
  \norm{\gamma(A')} \leq \norm{\pi(\Xi(A))} \leq \norm{\Xi(A)} 
  \leq \norm{\gamma(A)},
\end{equation}
where $\norm{x}$ represents the volume of the object $x$ (often called the 
area for a \mbox{codimension-2} surface in spacetime).  The inequalities 
arise from minimization, projection, and maximization, respectively.

\subsubsection{Subregion Flow Contained within Entanglement Wedge}

Suppose one were given access to a finite subregion $A$ on the leaf 
$\sigma$ and chose to apply the holographic slice construction only to 
this subregion. The result would be a sequence of renormalized leaves 
given by $\sigma(\lambda) = A(\lambda) \cup \overline{A}$, with 
$A(\lambda)$ denoting the sequence of subregions that result from 
radially evolving $A$ as illustrated in Fig.~\ref{fig:subregion}.
\begin{figure}[t]
\begin{center}
  \includegraphics[scale=0.3]{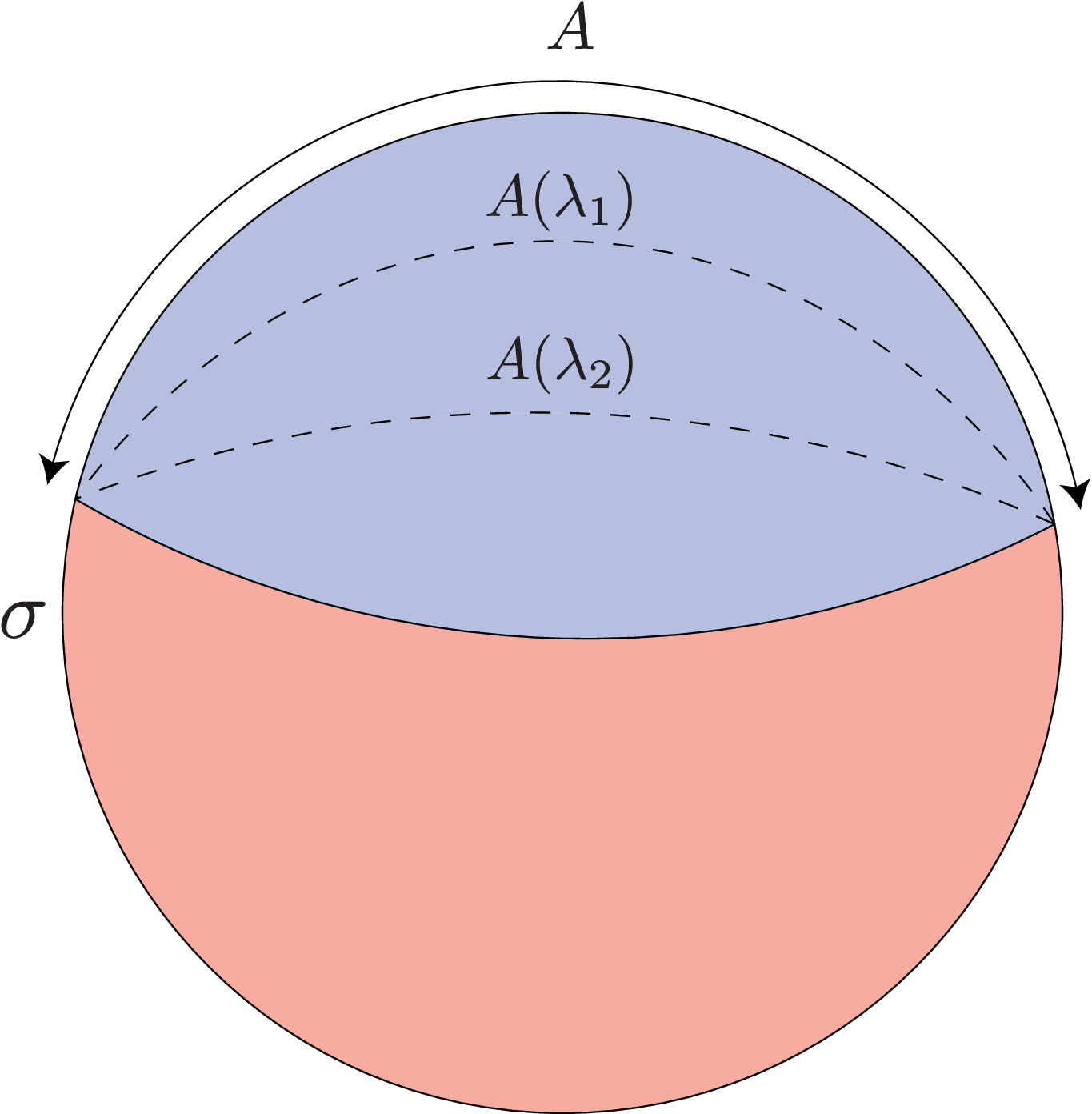}  
\end{center}
\caption{The radial evolution procedure when restricted to a subregion $A$ 
 results in a new leaf $\sigma(\lambda) = A(\lambda) \cup \overline{A}$, 
 where $A$ is mapped to a subregion $A(\lambda)$ contained within 
 $\text{EW}(A)$ (blue).  The figure illustrates this for two values 
 of $\lambda$ with $\lambda_2 < \lambda_1 < 0$ (dashed lines).}
\label{fig:subregion}
\end{figure}

An interpretation of this procedure as coarse-graining requires that it 
should not add any further information than what was already available. 
Thus, it should not allow one to reconstruct points in the bulk beyond 
what was already accessible from $A$, i.e.\ $\text{EW}(A)$.  This is 
ensured by the fact that the boundary of $\text{EW}(A)$ acts as an extremal 
surface barrier for HRRT surfaces anchored to points inside $\text{EW}(A)$, 
and thus, at no step does $A(\lambda)$ cross outside $\text{EW}(A)$.  In 
fact, if there were a non-minimal extremal surface anchored to $A$ which 
is contained within $\text{EW}(A)$, the holographic slice would not be able 
to go beyond this.  This would be the case if one were to consider $A$ to 
be a large subregion of the boundary dual to an AdS black hole.

\subsubsection{Probes Directly Reconstructable Region}

Using the definition of Ref.~\cite{Nomura:2017npr}, we define the directly 
reconstructable region of spacetime as the set of bulk points which can 
be localized by the intersections of some set of boundary anchored HRRT 
surfaces and the boundary of their entanglement wedges.  Boundary operators 
corresponding to the maximally localizable bulk operators are dual to local 
operators in the directly reconstructable region~\cite{Sanches:2017xhn}.

As argued in Ref.~\cite{Nomura:2017npr}, the interior of an equilibrated 
black hole cannot be reconstructed using the intersection of entanglement 
wedges.  Since the horizon acts as a barrier for all extremal surfaces 
anchored to points outside the black hole, $\sigma(\lambda)$ stays outside 
the horizon at each step.  Thus, the holographic slice cannot enter the 
black hole interior.  This implies that bulk regions that are not directly 
reconstructable using the entire holographic screen are inaccessible to 
any holographic slice.

In fact, as long as $\sigma(\lambda)$ does not become extremal, the 
holographic slice procedure can continue moving the leaf spatially inward. 
This is a consequence of Theorem~1 in Ref.~\cite{Nomura:2017fyh}.  As we 
will see later, the radial evolution will only halt once the boundary state 
on the renormalized leaf no longer has distillable local correlations. 
This can happen in two ways.  The first is that the surface closes off 
to zero area (corresponding to the vanishing of the coarse-grained Hilbert 
space).  The second is if the surface asymptotes to a bifurcation surface 
or Killing horizon (corresponding to a maximally entangled state).

Holographic slices probe entanglement shadows, the spacetime regions which 
cannot be probed by HRRT surfaces anchored to a non-renormalized leaf 
$\sigma$.  The extremal surfaces anchored to $\sigma$ that probe the 
shadow regions are non-minimal.  This prevents the reconstruction of 
points in these regions by using intersection of HRRT surfaces anchored 
to $\sigma$.  However, the set of HRRT surfaces used for constructing 
subsequent renormalized leaves need not be minimal on $\sigma$; they 
need only be minimal on the renormalized leaf at hand.  Since non-minimal 
surfaces anchored to $\sigma$ become minimal when anchored to an 
appropriately renormalized leaf $\sigma(\lambda)$, the holographic 
slice can flow through entanglement shadows.  We will see an example 
of this in the next section. 

Not only will the holographic slice itself flow through entanglement 
shadows, HRRT surfaces anchored to a renormalized leaf will probe regions 
behind shadows of the original leaf.   Again, this is because the minimal 
extremal surfaces anchored to renormalized leaves need not be portions 
of minimal extremal surface anchored to the original leaf.  In fact, 
one immediately starts recovering portions of shadow regions once the 
boundary is pulled in.  This is what occurs in dense stars and conical 
AdS.  In these cases, more and more of the shadow is recovered as the 
boundary is pulled in, and the entire shadow is only recovered once the 
slice contracts to a point.%
\footnote{This will be explored in future work.  This is similar 
 to how entanglement of purification probes portions of shadow 
 regions~\cite{Espindola:2018ozt}.  However, to recover the entire 
 shadow region using entanglement of purification, one must impose 
 additional conditions on the purification~\cite{Bao:2017nhh,NB}.}

Because the holographic slices pass through shadows, it seems that 
the collection of all holographic slices anchored to boundary leaves 
will sweep out the directly reconstructable region.  However, this 
breaks down if one considers a disconnected boundary.  Consider the 
case of a two-sided AdS black hole.  The directly reconstructable 
region can include regions behind the horizon if one picks a foliation 
of the left and right boundaries with an offset in time (see Appendix~A 
of Ref.~\cite{Nomura:2017fyh}).  The intersection of HRRT surfaces 
anchored to large subregions with support on both boundaries allow 
for the reconstruction of points behind the horizon.  On the other 
hand, the holographic slice is built from infinitesimal HRRT surfaces 
anchored to individual boundaries and hence cannot recover regions 
built from these long range correlations.  In the two-sided black hole 
(no matter what the offset in boundary times), the holographic slice 
will always connect through the bifurcation surface and never probe 
behind the horizon; Section~\ref{subsec:BH} explains this in detail.

\section{Examples}
\label{sec:example}

In this section, we illustrate salient properties of the holographic slice 
using a few example spacetimes.

\subsection{Conical AdS}
\label{subsec:conical}

We first consider conical $\text{AdS}_3$ to illustrate that the holographic 
slice probes regions inside entanglement shadows.  In order to obtain 
conical $\text{AdS}_3$, we start with the $\text{AdS}_3$ metric
\begin{equation}
  ds^2 = -\biggl( 1 + \frac{r^2}{L^2} \biggr) dt^2 
    + \biggl( 1 + \frac{r^2}{L^2} \biggr)^{-1} dr^2 + r^2 d\theta^2,
\label{eq:cAdS}
\end{equation}
where $L$ is the AdS length scale.  We then perform a $\mathbb{Z}_n$ 
quotient, so that the angular coordinate $\theta$ has periodicity 
$2\pi/n$.  Locally, this spacetime is identical to $\text{AdS}_{3}$, 
and solves the Einstein equations for a negative cosmological constant 
away from $r=0$.  However, there is a conical defect at $r=0$ introduced 
by the $\mathbb{Z}_n$ quotient.

\begin{figure}[t]
  \includegraphics[scale=0.7]{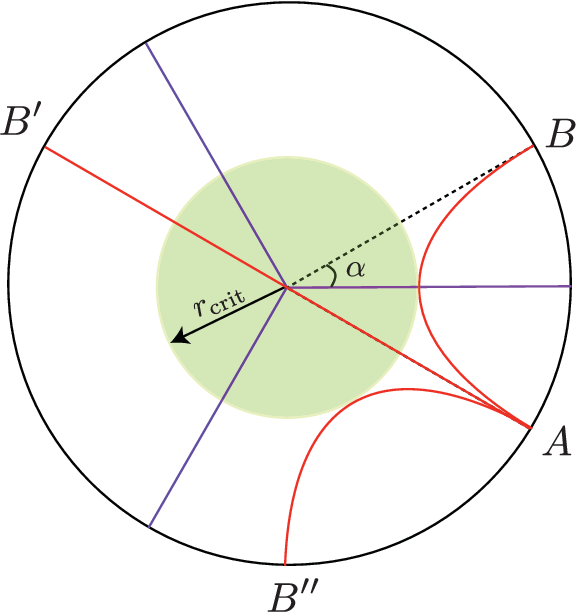} 
\caption{The case of conical $\text{AdS}_3$ with $n=3$.  The points $B$, 
 $B'$, and $B''$ are identified.  There are 3 geodesics from $A$ to $B$, 
 of which generically only one is minimal.  Here, we have illustrated 
 the subregion $AB$ with $\alpha = \pi/6$, where two of the geodesics 
 are degenerate.  This is the case in which the HRRT surface probes deepest 
 into the bulk, leaving a shadow region in the center.  Nevertheless, the 
 holographic slice spans the entire spatial slice depicted.}
\label{fig:cAdS}
\end{figure}
HRRT surfaces in this spacetime simply correspond to minimal length 
geodesics anchored to subregions at the conformal boundary.  As illustrated 
in Fig.~\ref{fig:cAdS}, there are $n$ geodesics in the parent $\text{AdS}_3$ 
spacetime which are candidate extremal surfaces for a given subregion. 
However, generically only one of them is minimal and corresponds to the 
HRRT surface. The geodesics in $\text{AdS}_3$ are described by the equation
\begin{equation}
  \tan^2\!\theta = \frac{r^2 \tan^2\!\alpha -L^2}{r^2 + L^2},
\label{eq:con-geodesic}
\end{equation}
where $\alpha$ is the half-opening angle of the subregion being considered. 
Since the angular coordinate $\theta$ has a periodicity of $2\pi/n$, 
the minimal length geodesic that probes deepest into the bulk is obtained 
when $\alpha = \pi/2n$.  From Eq.~(\ref{eq:con-geodesic}), this gives 
a critical radius of~\cite{Balasubramanian:2014sra}
\begin{equation}
  r_{\text{crit}}(n) = L \cot\left(\frac{\pi}{2n}\right),
\label{eq:rcrit}
\end{equation}
which takes a nonzero finite value for $n \neq 1$.  Thus, the region 
$r < r_{\text{crit}}(n)$ is an entanglement shadow, which cannot be probed 
by HRRT surfaces anchored to the conformal boundary.

The holographic slice is constructed by finding infinitesimal HRRT surfaces 
starting from the (regularized) conformal boundary.  Since the spacetime 
is locally $\text{AdS}_3$, the HRRT surfaces for small regions are identical 
to those in $\text{AdS}_3$.  Because of the static and spherically symmetric 
nature of $\text{AdS}_{3}$, the renormalized leaves correspond to surfaces 
of constant $r$ and $t$.  Now, since $r_{\text{crit}}(n)$ is not an extremal 
surface barrier, as can be seen from the existence of non-minimal extremal 
surfaces penetrating it, the holographic slice suffers no obstruction in 
crossing over to the entanglement shadow.  This implies that the holographic 
slice is simply given by a constant time slice that covers all of the 
spatial region $r \in [0,\infty)$.

In general, holographic slices do not have any difficulty in going into 
entanglement shadow regions, since these shadows are not associated with 
extremal surface barriers which any extremal surfaces anchored to the 
outside cannot penetrate~\cite{Engelhardt:2013tra}.  In fact, holographic 
slices also sweep entanglement shadows other than those in the centers of 
conical AdS, e.g.\ regions around a dense star~\cite{Freivogel:2014lja}.

\subsection{Black Holes}
\label{subsec:BH}

Consider a two-sided eternal $\text{AdS}_{d+1}$ Schwarzschild black hole. 
The metric is given by
\begin{equation}
  ds^2 = -f(r)\, dt^2 + \frac{1}{f(r)} dr^2 + r^2 d\Omega_{d-1}^{2},
\label{eq:eternal}
\end{equation}
where
\begin{equation}
  f(r) = 1 + \frac{r^2}{L^2} -\Bigl( \frac{r_+}{r} \Bigr)^{d-2} 
    \biggl( 1 + \frac{r_+^2}{L^2} \biggr),
\label{eq:fr}
\end{equation}
with $r_+$ being the horizon radius.  As can be seen in Fig.~\ref{fig:BH}, 
the two exterior regions have a timelike Killing vector.  Thus, the HRRT 
surfaces anchored to subregions with support only on one boundary respect 
the Killing symmetry and lie on the constant $t$ slice connecting the 
respective boundary to the bifurcation surface.  Subregions anchored 
on both boundaries could potentially lie on a different spatial slice 
if the HRRT surface is connected.  However, since we are considering 
the holographic slice being built up using HRRT surfaces anchored to 
infinitesimal subregions, those anchored on both sides always stay 
disconnected.
\begin{figure}[t]
  \includegraphics[scale=0.7]{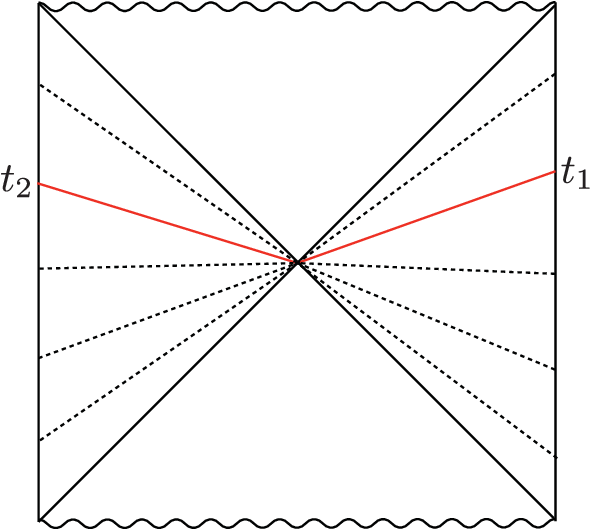} 
\caption{The exterior of a two-sided eternal AdS black hole can be 
 foliated by static slices (black dotted lines).  The holographic slice 
 (red) connects the boundary time slices at $t = t_1$ on the right 
 boundary and $t = t_2$ on the left boundary to the bifurcation surface 
 along these static slices.}
\label{fig:BH}
\end{figure}

Thus, the holographic slice, as seen in Fig.~\ref{fig:BH}, is the union 
of static slices in both exterior regions and terminates at the bifurcation 
surface.  As shown in Ref.~\cite{Nomura:2017fyh}, the bifurcation surface 
itself is extremal and lies on a Killing horizon, and hence the process 
of renormalizing leaves must asymptote to this surface.

The phenomenon of a holographic screen being terminated at a nontrivial 
surface requires the existence of a bifurcation surface, which is absent 
in most physical situations.  For example, consider an AdS-Vaidya metric 
where a black hole is formed from the collapse of a thin null shell of 
energy~\cite{Hubeny:2013dea}.  The metric in ingoing Eddington-Finkelstein 
coordinates is given by
\begin{equation}
  ds^2 = -f(r,v)\, dv^2 + 2dv\, dr + r^2 d\Omega_{d-1}^2,
\label{eq:AdS-V}
\end{equation}
where
\begin{equation}
  f(r,v) = 1 + \frac{r^2}{L^2} - \theta(v)\, 
    \Bigl( \frac{r_+}{r} \Bigr)^{d-2} \biggl( 1 + \frac{r_+^2}{L^2} \biggr),
\label{eq:f_rv}
\end{equation}
with
\begin{equation}
  \theta\left(v\right) = \begin{cases}
    0 & \text{for }\, v < 0 \\
    1 & \text{for }\, v > 0.
  \end{cases}
\label{eq:theta_v}
\end{equation}
The null shell lies at $v = 0$, and this spacetime is obtained simply by 
stitching together an AdS-Schwarzschild metric to the future of the shell 
and a pure AdS metric to the past.  The composite global spacetime is 
time-dependent, but each of the building blocks admits a timelike Killing 
vector locally as shown in Fig.~\ref{fig:Vaidya}.  As discussed earlier, 
since the HRRT surfaces relevant to the holographic slice are those of 
infinitesimal subregions, they only sense the local spacetime, which is 
static.  This allows us to construct the holographic slice independently 
in each region.  The static slices can then be stitched together to obtain 
the holographic slice as shown in Fig.~\ref{fig:Vaidya}.
\begin{figure}[t]
  \includegraphics[scale=0.7]{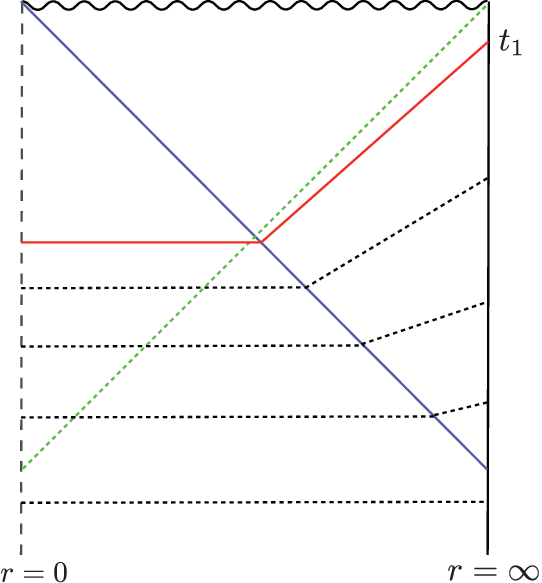}
\caption{Penrose diagram of an AdS Vaidya spacetime formed from the 
 collapse of a null shell (blue), resulting in the formation of an 
 event horizon (green).  Individual portions of the spacetime, the 
 future and past of the null shell, are static.  Thus, the holographic 
 slice (red) can be constructed by stitching together a static slice 
 in each portion.}
\label{fig:Vaidya}
\end{figure}

An important feature here is that at late times, i.e.\ sufficiently after 
the black hole has stabilized, the holographic slice constructed from 
a leaf stays near the horizon for long time.  Eventually, this flow 
terminates at $r = 0$.  This behavior of holographic slices is, in fact, 
general in one-sided black holes; see Fig~\ref{fig:E-F} for a schematic 
depiction.
\begin{figure}[t]
  \includegraphics[scale=0.6]{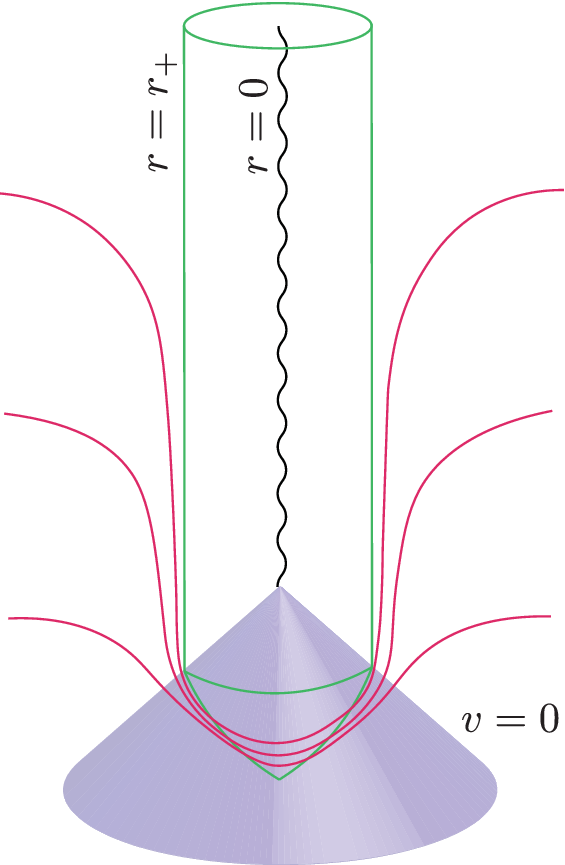}
\caption{A schematic depiction of holographic slices for a spacetime with 
 a collapse-formed black hole in ingoing Eddington-Finkelstein coordinates.}
\label{fig:E-F}
\end{figure}

Note that the picture of Fig~\ref{fig:E-F} is obtained in the $G_{\rm N} 
\rightarrow 0$ limit.  When $G_{\rm N} \neq 0$, renormalized leaves will 
hit the stretched horizon~\cite{Susskind:1993if}, where the semiclassical 
description of spacetime breaks down, before being subjected to the long 
flow near the horizon.

\subsection{FRW Spacetimes}
\label{subsec:FRW}

We now discuss a nontrivial example in a time-dependent spacetime, away 
from the standard asymptotically AdS context.  Consider a $(d+1)$-dimensional 
flat Friedmann-Robertson-Walker (FRW) spacetime containing a single fluid 
component with equation of state parameter $w$
\begin{equation}
  ds^2 = a(\eta)^2 \left( -d\eta^2 + dr^2 + r^2 d\Omega_{d-1}^2 \right).
\label{eq:FRW}
\end{equation}
Here,
\begin{equation}
  a(\eta) = c \,|\eta|^q,
\label{eq:FRWa}
\end{equation}
where $c > 0$ is a constant, and
\begin{equation}
  q = \frac{2}{d - 2 + d w}.
\label{eq:FRWq}
\end{equation}
The discontinuity of $q$ at $w = (2-d)/d$ is an artifact of choosing 
conformal time, and physics is smooth across this value of $w$.

The spherically symmetric holographic screen is located at
\begin{equation}
  r(\eta) = \frac{a(\eta)}{\frac{da}{d\eta}(\eta)} = \frac{\eta}{q}.
\label{eq:FRWscreen}
\end{equation}
By spherical symmetry, the holographic slice must be a \mbox{codimension-1} 
surface of the form $\eta = g(r)$, where each renormalized leaf 
is an $\mathbb{S}^{d-1}$.  Consider a renormalized leaf at 
$\eta = \eta_*$ and $r = r_*$.  Generalizing the results from 
Refs.~\cite{Nomura:2017fyh,Nomura:2017grg}, HRRT surfaces anchored 
to a small spherical cap of half-opening angle $\gamma$ of the renormalized 
leaf are given by
\begin{equation}
  \eta(\xi) = \eta_* + \frac{\dot{a}}{2a} 
    \bigl( \xi_*^2-\xi^2 \bigr) + \dots,
\label{eq:FRW-HRRT}
\end{equation}
where $\xi = r \sin\theta$ and $\xi_* = r_* \sin\gamma$ with 
$\theta$ being the polar angle, and $a \equiv a(\eta_*)$ and 
$\dot{a} \equiv da/d\eta(\eta_*)$.  We refer the reader to 
Appendix~C.3 of Ref.~\cite{Nomura:2017fyh} for more details.

The next renormalized leaf is generated by joining together the deepest 
point of each such HRRT surface.  Suppose $\varDelta \eta$ and $\varDelta r$ 
represent the change in conformal time and radius from one renormalized 
leaf to the next.  Then we have
\begin{align}
  \varDelta \eta &= \frac{\dot{a}}{2a} \xi_*^2 + \dots,
\label{eq:Delta_eta}\\
  \varDelta r &= -( r_* - r_* \cos\gamma ),
\label{eq:Delta_r}
\end{align}
so that
\begin{equation}
  \frac{\varDelta \eta}{\varDelta r} 
  = -\frac{\frac{\dot{a}}{2a} r_*^2 \sin^2\!\gamma}{r_* (1 - \cos\gamma)} 
  =-\frac{\dot{a}}{a} r_*.
\label{eq:Del_eta-Del_r}
\end{equation}
Taking the limit $\gamma \rightarrow 0$, we obtain a differential equation 
for the radial evolution of the holographic slice
\begin{equation}
  \frac{d\eta}{dr} = -\frac{q r}{\eta}.
\label{eq:FRW-evol}
\end{equation}
Integrating this equation gives us
\begin{equation}
  \eta^2 + q r^2 = \eta_*^2 + q r_*^2 
  = \frac{1+q}{q} \eta_0^2,
\label{eq:FRW-slice}
\end{equation}
where $\eta_0$ is the conformal time of the original non-renormalized leaf.

Let us highlight a few interesting features of this holographic slice. 
First, it spans the entire interior region of the holographic screen. 
Next, substituting $r = \eta/q$ into Eq.~(\ref{eq:FRW-evol}) tells us 
that the holographic slice starts out in the null direction from the leaf; 
see Eq.~(\ref{eq:FRWscreen}).  This is because the $k$ direction locally 
has zero expansion there, as discussed in Section~\ref{subsec:prop}. 
As we move inward along the radial flow, however, the slope becomes flatter, 
and eventually the surface reaches the highest point given by
\begin{equation}
  \eta(r=0) = \eta_0 \sqrt{\frac{1+q}{q}}.
\label{eq:FRW-top}
\end{equation}
In Fig.~\ref{fig:FRW} we have depicted holographic slices, given by 
Eq.~(\ref{eq:FRW-slice}), for several values of $w$ with $d = 3$.
\begin{figure}[t]
  \includegraphics[scale=0.75]{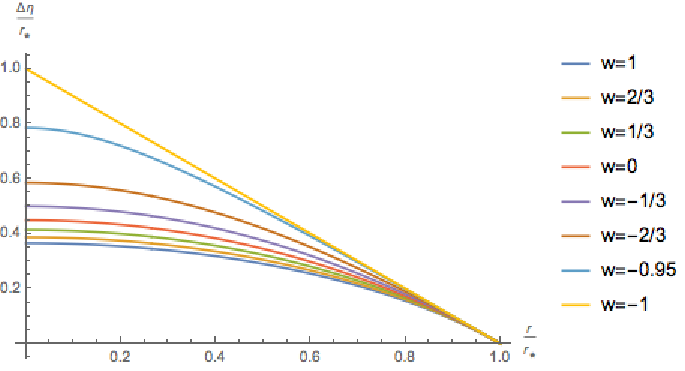}
\caption{Holographic slices of $(3+1)$-dimensional flat FRW universes 
 containing a single fluid component with equation of state parameter $w$.}
\label{fig:FRW}
\end{figure}

\subsection{Asymptotically AdS and Flat Spacetimes}
\label{subsec:AlAdS}

Here we discuss certain subtleties associated with holographic screens 
that lie on an asymptotic boundary.  First, consider a $(d+1)$-dimensional 
asymptotically AdS spacetime, which can be expanded in a Fefferman-Graham 
series~\cite{FG} near the boundary
\begin{equation}
  ds^2 = \frac{L^2}{z^2} \bigl\{ g_{ab}(x^a,z) dx^a dx^b + dz^2 \bigr\},
\label{eq:FG-1}
\end{equation}
where $L$ is the AdS length scale, and
\begin{equation}
  g_{ab}(x^a,z) = g_{ab}^{(0)}(x^a) + z^2 g_{ab}^{(2)}(x^a) + \dots.
\label{eq:FG-2}
\end{equation}
Here, $g_{ab}^{(0)}$ represents the conformal boundary metric, and the 
subleading corrections represent deviations as one moves away from the 
boundary at $z=0$.

In an asymptotically AdS spacetime, the holographic screen $H$ formally 
lies at spacelike infinity.  In order to construct a holographic slice 
in such a situation, one needs to first consider a regularized screen 
$H'$ at $z = \epsilon$ and then take the limit $\epsilon \rightarrow 0$ 
after constructing the slice.  Suppose that a leaf is given by a constant 
$t$ slice of $H'$, with $h_{ij}$ representing the induced metric on the 
leaf.  The null normals are then given by
\begin{align}
  k_\mu &= dz - dt + O(\epsilon^2),
\label{eq:AdS-normal-1}\\
  l_\mu &= -dz - dt +O(\epsilon^2),
\label{eq:AdS-normal-2}
\end{align}
where the $O(\epsilon^2)$ corrections arise due to deviations away from 
the boundary.  The null expansions are
\begin{align}
  \theta_k &= h^{ij} \Gamma^z_{ij} - h^{ij} \Gamma^t_{ij} 
\nonumber\\
  &= -\frac{\epsilon(d-1)}{L} + O(\epsilon^2),
\label{eq:AdS-exp-1}\\
  \theta_l &= -h^{ij} \Gamma^z_{ij} - h^{ij} \Gamma^t_{ij} 
\nonumber\\
  &= \frac{\epsilon(d-1)}{L} + O(\epsilon^2).
\label{eq:AdS-exp-2}
\end{align}
Thus, we see that the expansion $\theta_k$ vanishes only in the limit 
$\epsilon \rightarrow 0$.

This implies that a leaf $\sigma'$ of a regularized screen $H'$ 
($\epsilon \neq 0$) is, in fact, a renormalized leaf ($\theta_k \neq 0$), 
and thus the results in Section~\ref{subsec:prop}---that the holographic 
slice initially extends in the null direction and the leaf area remains 
constant---do not apply.  In fact, the holographic slice extending from 
$\sigma'$ initially evolves inward along the $z$ direction up to corrections 
of $O(\epsilon)$, as can be seen from the fact that $\theta_k = -\theta_l$ 
up to $O(\epsilon)$.  In the limit $\epsilon \rightarrow 0$, both $\theta_k$ 
and $\theta_l$ vanish simultaneously.  This leads to a holographic screen 
at spacelike infinity in a formal sense.%
\footnote{Strictly speaking, this does not satisfy the definition of 
 the holographic screen in Section~\ref{subsec:def}, which requires 
 $\theta_l$ to be strictly positive.}

A similar situation arises in asymptotically flat spacetimes.  A general 
asymptotically flat spacetime can be expanded in the Bondi-Sachs 
form~\cite{Bondi:1962px,Sachs:1962wk} as
\begin{align}
  ds^2 =& -\frac{V}{r} e^{2\beta} du^2 - 2 e^{2\beta} du dr
\nonumber\\
  & {} + r^2 h_{AB} \bigl(dx^A - U^A du\bigr) \bigl(dx^B - U^B du\bigr),
\label{eq:AlFlat}
\end{align}
where each of the functions admits a large $r$ expansion with the 
following behavior:
\begin{alignat}{3}
  & V = r + O(1),\qquad & \beta = O\biggl(\frac{1}{r^2}\biggr),
\label{eq:Bondi-1}\\
  & U^A = O\biggl(\frac{1}{r^2}\biggr), \qquad & h_{AB} = O(1).
\label{eq:Bondi-2}
\end{alignat}

In order to construct a holographic slice, the holographic screen $H$ 
must be regularized to become a timelike surface $H'$ at $r = R$, where 
we can eventually take the limit $R \rightarrow \infty$.  The null normals 
of a leaf on a constant time slice are
\begin{align}
  k_\mu &= du,
\label{eq:Flat-normal-1}\\
  l_\mu &= -\frac{V}{r} du - 2 dr,
\label{eq:Flat-normal-2}
\end{align}
giving the null expansions near the boundary
\begin{align}
  \theta_k &= -\frac{2}{R} + O\biggl(\frac{1}{R^2}\biggr),
\label{eq:Flat-exp-1}\\
  \theta_l &= \frac{2}{R} + O\biggl(\frac{1}{R^2}\biggr).
\label{eq:Flat-exp-2}
\end{align}
Thus, similar to the case of asymptotically AdS spacetimes, a leaf of a 
regularized holographic screen is a renormalized one, and both $\theta_k, 
\theta_l \rightarrow 0$ simultaneously as $R \rightarrow \infty$.

As a simple example, we illustrate the case of a Minkowski spacetime 
in Fig.~\ref{fig:Minkowski}.  As the limit $R \rightarrow \infty$ is 
taken, the holographic slices become complete Cauchy hypersurfaces which 
are constant time slices anchored to spatial infinity.  In the limit 
$t \rightarrow +\infty$ ($-\infty$), future (past) null and timelike 
infinities are obtained as a holographic slice.  In this situation, 
time evolution of the boundary theory from $t \rightarrow -\infty$ 
to $+\infty$ corresponds to an $S$-matrix description of the bulk.
\begin{figure}
  \includegraphics[scale=0.65]{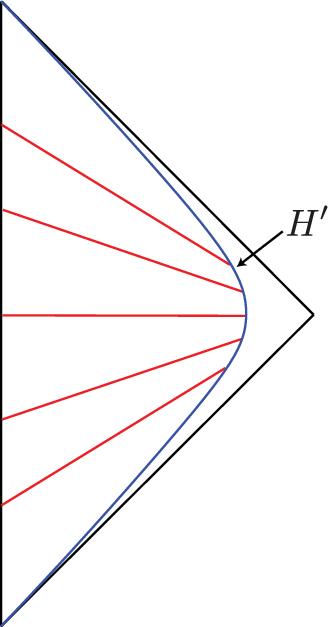}
\caption{Penrose diagram of a Minkowski spacetime.  The holographic 
 slices (red) are anchored to the regularized holographic screen $H'$ 
 (blue).  As the limit $R \rightarrow \infty$ is taken, the holographic 
 slices become complete Cauchy slices.}
\label{fig:Minkowski}
\end{figure}

\section{Interpretation and Applications}
\label{sec:interpret}

We have introduced the geometric definition of the holographic slice and 
demonstrated some of its properties.  But what does the slice correspond 
to in the boundary theory?  What questions can it help us address?  The 
construction of the slice naturally lends itself to an interpretation 
of eliminating information at small scales, and hence can be well 
understood in the context of coarse-graining.  Through this, we can 
think of the slice as an isometric tensor network.  This provides us 
with a new way to think about holographic tensor networks and the bulk 
regions of spacetime that they encode.

Throughout this section we will be talking about various Hilbert 
spaces in which holographic states belong.  To do so, we will be taking 
$G_{\rm N}$ to be finite but small.  This is an appropriate approximation 
for classical spacetimes provided we only concern ourselves with length 
scales sufficiently larger than the Planck length.

\subsection{Coarse-Graining}
\label{subsec:coarse}

We take the view that a boundary state, $\ket{\psi(0)}$, lives on the 
original leaf, $\Upsilon(0)$, i.e.\ it lives in an effective Hilbert 
space, $\mathcal{H}$, having a local product space structure with 
dimension $\log{|\mathcal{H}|} = \norm{\Upsilon(0)}/4 G_{\rm N}$. 
The HRRT prescription says that the emergent bulk geometry is intricately 
related to the entanglement of the boundary state.  In particular, 
despite the fact that bulk information is delocalized in the boundary 
theory, a bulk region cannot be reconstructed if some boundary subregions 
are ignored.  The size of the smallest subregion for this to occur 
gives us some idea of what scale of boundary physics this bulk region 
is encoded in.  Using this intuition, we can then attempt to address 
what coarse-graining the boundary state corresponds to in the bulk.

At each step in the construction of the holographic slice, we eliminate 
the region of spacetime associated with a small length scale, $\delta$, 
of the boundary.  In particular, this is the region of spacetime whose 
information is necessarily lost if we cannot resolve below length scale 
$\delta$.  In this sense, we are coarse-graining over the scale $\delta$ 
and obtaining a new bulk region whose information has not been lost. 
Recursively doing this and sending $\delta$ to zero produces a continuum 
of bulk domains of dependence with unique boundaries sweeping out 
the holographic slice, $\Upsilon(\lambda)$.  This is depicted in 
Fig.~\ref{fig:RG2}.
\begin{figure}[t]
  \includegraphics[scale=0.25]{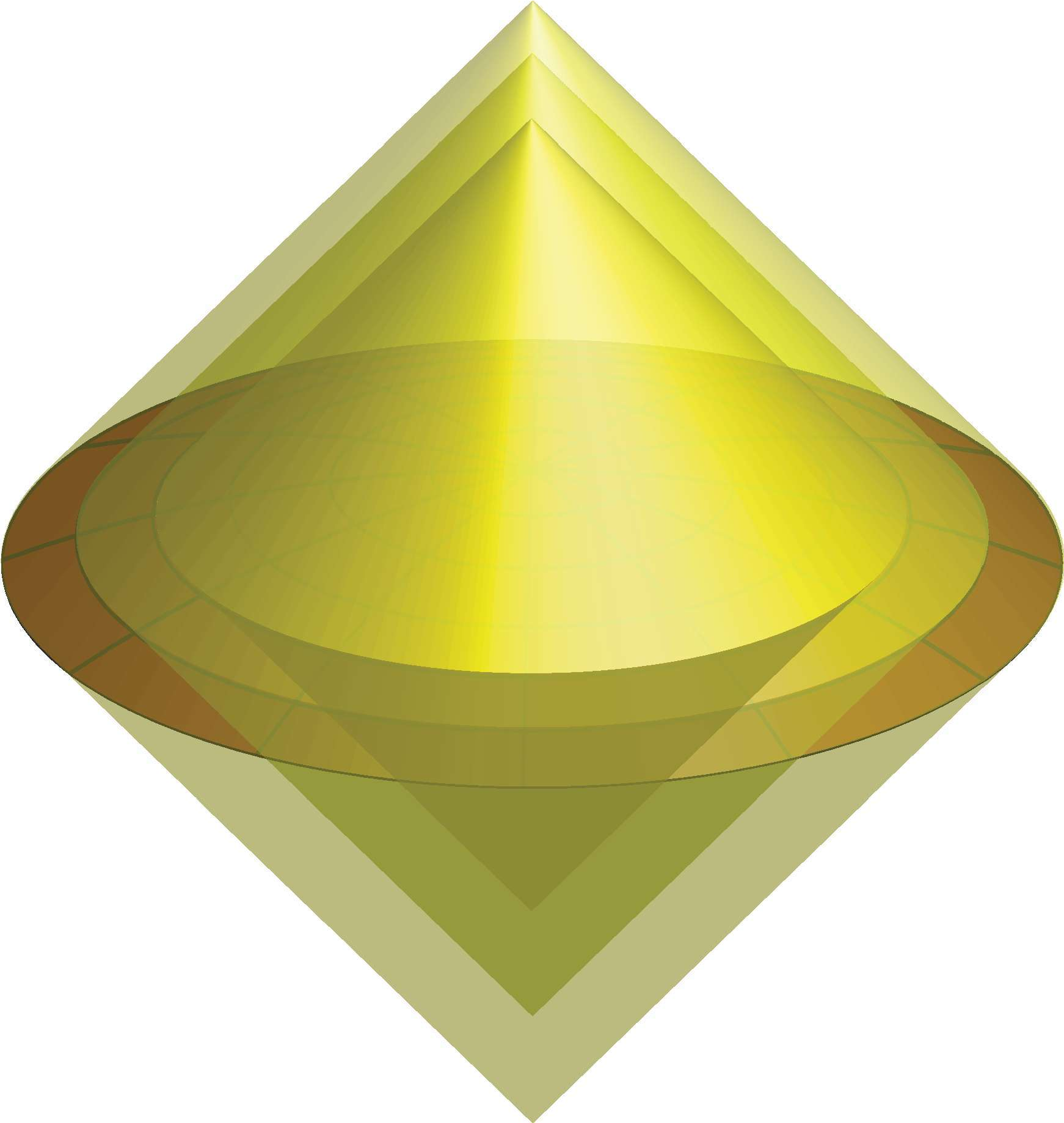}
\caption{This depicts the holographic slice (maroon), and the successive 
 domains of dependence encoded on each renormalized leaf.}
\label{fig:RG2}
\end{figure}

A consistency check for this interpretation is that the size of the 
effective Hilbert space should necessarily decrease as we coarse-grain 
over larger and larger scales.  This is precisely the monotonicity 
property listed in Section~\ref{sec:holo_slice}:\ as one flows along 
the holographic slice, the area of the renormalized leaves decreases. 
This tells us that the size of the effective Hilbert space describing 
the bulk domain of dependence also decreases.

As the coarse-graining procedure progressively removes information 
at small scales, a corresponding removal of bulk information closest 
to the renormalized leaf occurs.  Given this global removal of short 
range information, one should expect the entanglement between any region 
and its complement to correspondingly decrease.  This is precisely 
the monotonicity of entanglement entropy property observed in 
Section~\ref{sec:holo_slice}.  This is consistent with the interpretation 
that at each step we are removing short range entanglement.

Using the holographic slice, we can address the question of how much 
entanglement between a subregion and its complement is sourced by physics 
at different scales.  By following the integral curves of $s$ for the 
subregion, we can stop at whatever scale we desire and use the HRRT 
prescription on the renormalized leaf.  This gives us the entanglement 
entropy sourced by physics at length scales larger than that associated 
with the renormalized leaf.

\subsection{Radial Evolution of States}
\label{subsec:rs-corresp}

We will now be more explicit in describing the framework for coarse-graining 
holographic states.  Given a bulk region $D(\Upsilon(0))$, there exists 
a quantum state $\ket{\psi(0)}$ living in some fundamental holographic 
Hilbert space, $\mathcal{H}_{\rm UV}$, in which the bulk information 
of $D(\Upsilon(0))$ is encoded via the HRRT prescription.  This implies 
that $\mathcal{H}_{\rm UV}$ has a locally factorizable structure.  On 
the other hand, the area of $\Upsilon(0)$ provides an upper bound for 
the dimension of effective Hilbert space that $\ket{\psi(0)}$ lives in, 
which we call $\mathcal{H}_{\Upsilon(0)}$.  That is, $\ket{\psi(0)} \in 
\mathcal{H}_{\Upsilon(0)} \subset \mathcal{H}_{\rm UV}$.  The dimension, 
$|\mathcal{H}_{\Upsilon(0)}|$, of the effective Hilbert space is defined as
\begin{equation}
  \ln{|\mathcal{H}_{\Upsilon(0)}|} 
  = \sum_i S_i = \frac{\norm{\Upsilon(0)}}{4 G_{\rm N}}.
\label{eq:eff-dim}
\end{equation}
Here, $S_i$ represents the entanglement entropy of $\ket{\psi(0)}$ in an 
infinitesimally small subregion, $A_i$, of the holographic space $\Omega$ 
on which $\mathcal{H}_{\rm UV}$ is defined.  We sum over all of these small 
subregions such that $\Omega = \cup_i A_i$ and $A_i \cap A_j = \emptyset$ 
($i \neq j$).  This reduces to calculating the area of $\Upsilon(0)$ 
because of the HRRT prescription.  Namely, the size of the effective 
Hilbert space that $\ket{\psi(0)}$ lives in is determined by the entanglement 
between the fundamental degrees of freedom of $\mathcal{H}_{\rm UV}$, and 
$\norm{\Upsilon(0)}/4 G_{\rm N}$ is the thermodynamic entropy associated 
with this entanglement structure.

As one continuously coarse-grains $\ket{\psi(0)}$, information encoded in 
small scales is lost.  Correspondingly, information of the bulk geometry 
that is stored in small scales is lost, and the dimension of the 
effective Hilbert space that the coarse-grained state lives in decreases. 
At a given scale of coarse-graining corresponding to $\lambda$, the 
new state, $\ket{\psi(\lambda)}$, lives in the same Hilbert space as 
the original leaf, $\mathcal{H}_{\rm UV}$, but now in an effective 
subspace, $\mathcal{H}_{\Upsilon(\lambda)}$, with dimension given by 
$\ln{|\mathcal{H}_{\Upsilon(\lambda)}|} = \norm{\Upsilon(\lambda)}/4 
G_{\rm N}$.  Additionally, $\ket{\psi(\lambda)}$ only contains information of 
$D(\Upsilon(\lambda))$, as we have explicitly lost the information necessary 
to reconstruct any part of $D(\Upsilon(0)) \setminus D(\Upsilon(\lambda))$.

Because all of the coarse-grained states live in the same Hilbert space, 
$\mathcal{H}_{\rm UV}$, we can consider the coarse-graining procedure 
as a unitary operation that takes us from state to state, along the lines 
of the work of Ref.~\cite{Nozaki:2012zj}.  That is,
\begin{equation}
  \ket{\psi(\lambda)} = U(\lambda, 0) \ket{\psi(0)}.
\label{eq:unitary}
\end{equation}
We can write $U(\lambda_1, \lambda_2)$ as
\begin{equation}
  U(\lambda_1, \lambda_2) 
  = P \exp\left[ -i \int_{\lambda_2}^{\lambda_1} K(\lambda) d\lambda \right],
\label{eq:unitary-def}
\end{equation}
where $P$ represents path-ordering.  $K(\lambda)$ is a Hermitian operator 
removing physical correlations between nearby subregions at the length 
scale, $l_\lambda$, associated to $\lambda$ on $\Omega$.  Some appropriate 
measures of physical correlations would be the mutual information 
($I(A,B)$), entanglement negativity ($N(A,B)$), or the entanglement 
of purification ($E(A,B)$) between neighboring small subregions of the 
leaf under consideration, $\Upsilon(\lambda)$.  Note that as the boundary 
state becomes maximally entangled, all three of these measures will 
vanish.  This also happens if the boundary state in consideration has 
no entanglement, i.e.\ is a product state.

Of the three measures, entanglement of purification already has a bulk 
description that naturally characterizes some measure of moving into 
the bulk.  In particular, Refs.~\cite{Takayanagi:2017knl,Nguyen:2017yqw} 
proposed that the entanglement of purification of two boundary subregions, 
$A$ and $B$, is calculated by the minimum cross section, $\zeta$, 
of a bipartition of the extremal surface anchored to $A \cup B$; see 
Fig.~\ref{fig:purification}.  Considering the case where $\partial A$ 
and $\partial B$ coincide at some point, and whose connected phase is 
the appropriate extremal surface, the entanglement of purification gives 
some measure of the depth of the extremal surface.  In bulk dimensions 
higher than $2+1$, $\norm{\zeta}$ will not be in units of length, but 
it is still related to the depth of the extremal surface.  Thus it seems 
natural that $K(\lambda)$ is some \mbox{(quasi-)local} function, $F$, 
of physical correlations, including but not necessarily limited to 
quantum entanglement, at the scale $l_\lambda$.  For example, it may 
be related to the entanglement of purification:
\begin{equation}
  K(\lambda) \sim \int\! d^{d-1}{\bf x}\, 
    F\bigl( E_\lambda({\bf x}) \bigr),
\label{eq:K-lambda}
\end{equation}
where ${\bf x}$ are the coordinates of $\Omega$.  Here, $E_\lambda({\bf x})$ 
is the ``density'' of the entanglement of purification between the degrees 
of freedom in two neighboring regions on $\sigma(\lambda)$ around ${\bf x}$.
\begin{figure}[t]
\includegraphics[scale=1.2]{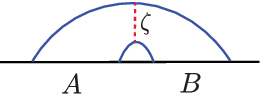}
\caption{Let $A$ and $B$ two boundary subregions.  The blue lines 
 represent the HRRT surface of $A \cup B$ and $\zeta$ the minimal cross 
 section.  The entanglement of purification of $A$ and $B$ is given by 
 $\norm{\zeta}/4 G_{\rm N}$.  In the limit that $A$ and $B$ share a boundary 
 point, $\zeta$ probes the depth of the extremal surface.}
\label{fig:purification}
\end{figure}

Alternately, the case of subregion coarse-graining as in 
Fig.~\ref{fig:subregion} motivates the usage of results from 
Ref.~\cite{Faulkner:2017vdd}, which can be used to map the state 
from $\sigma$ to $\sigma(\lambda)$.  In AdS/CFT, modular evolution 
allows one to explicitly reconstruct bulk operators on the HRRT surface. 
With our assumption that the HRRT formula holds (with quantum corrections), 
a similar construction should be possible given complete knowledge of 
the boundary theory.  $K(\lambda)$ may then be better understood as 
a convolution over modular evolutions with infinitesimal boundary 
subregions:
\begin{equation}
  K(\lambda) \sim \int\! d^{d-1}{\bf x}\, 
    F\bigl( \kappa_\lambda({\bf x}) \bigr),
\label{eq:K-modular}
\end{equation}
where $\kappa_\lambda({\bf x})$ is the modular Hamiltonian density on 
$\sigma(\lambda)$ at ${\bf x}$.  It would be interesting to make the 
connection of $K(\lambda)$ to modular evolution clearer in the future.

The process of removing short range correlations continues until all 
correlations at the scale $l_\lambda$ have been removed, and hence 
no more bulk spacetime can be reconstructed.  This can happen when the 
slice contracts to a point and no local product structure exists in the 
effective Hilbert space $\mathcal{H}_{\Upsilon(\lambda)}$.  Note that in 
$\mathcal{H}_{\rm UV}$ this state corresponds to a product state, so that 
$S_i = 0$ for every subregion in $\mathcal{H}_{\rm UV}$.  The other way 
in which all relevant correlations vanish is when the coarse-grained 
state becomes maximally entangled in $\mathcal{H}_{\Upsilon(\lambda)}$. 
In $\mathcal{H}_{\rm UV}$, this corresponds to a state which satisfies 
$S_A = \sum_{i \subset A} S_i$ for all subregions, $A$, of $\Omega$.

When the coarse-grained state becomes maximally entangled in 
$\mathcal{H}_{\Upsilon(\lambda)}$, $U(\lambda + d\lambda,\lambda)$ 
becomes the identity for Eq.~(\ref{eq:K-lambda}) as $K(\lambda)$ becomes 
$0$.  Hence, the state remains invariant under the coarse-graining 
operation.  There are two ways for this to happen geometrically.  One 
is if the slice approaches a bifurcation surface; then the extremal 
surfaces coincide with the renormalized leaf, hence preventing any 
further movement into the bulk.  This is the case for eternal two-sided 
black holes.  The second is if the slice approaches a null, non-expanding 
horizon and the state is identical along the horizon.  This is the 
case in de~Sitter space.  This result is complementary to Theorem~1 
of Ref.~\cite{Nomura:2017fyh}, which proves that if a boundary state 
is maximally entangled, it must either live on a bifurcation surface 
or a null, non-expanding horizon.

This may initially seem like a contradiction---that both the state becomes 
maximally entangled and that there are no more correlations to harvest. 
However, it is precisely because we are examining correlations at 
\textit{small} scales that this occurs.  A small boundary region is 
maximally entangled with the entire rest of the boundary, and hence 
the short range entanglement must vanish.  One can quantify this by 
examining the entanglement negativity of bipartitions of small subregions 
on $\sigma(\lambda)$.  As states become maximally entangled, the 
entanglement negativity vanishes for two small subregions.  This places 
an upper bound on the real, distillable entanglement between these 
subregions.  Hence, the true quantum entanglement at small scales vanishes 
as a state becomes maximally entangled.  Correspondingly, the coarse-graining 
procedure halts.  This is indeed what happens to the holographic slice.

The same can also be seen by considering the modular evolution, 
Eq.~(\ref{eq:K-modular}), of a maximally entangled subregion.  In this 
case, the modular evolution is proportional to the identity operator, 
and hence the modular flow of the state is stationary.  This corresponds 
to no movement into the bulk as expected by the properties of the 
holographic slice for maximally entangled states.

\subsection{Tensor Network Picture}
\label{subsec:tensor}

In this language the relationship to tensor networks is very clear. 
The holographic slice arises as the continuous limit of a tensor 
network that takes a boundary state and disentangles below a certain 
scale, reducing the effective Hilbert space size.  This is a slight 
generalization of continuous Multiscale Entanglement Renormalization 
Ansatz (cMERA)~\cite{Haegeman:2011uy}, which is restricted to hyperbolic 
geometries.

In general, we can consider a tensor network as a non-continuum modeling 
of the holographic slice, which isometrically embeds boundary states into 
spaces of lower effective dimension by removing short range correlations; 
see Fig.~\ref{fig:TNetwork}.  In the ground state of AdS/CFT, this 
corresponds to an instance of Multiscale Entanglement Renormalization 
Ansatz (MERA)~\cite{Vidal:2007hda}.  Each layer of the tensor network 
then lives on the corresponding renormalized leaf of the discrete version 
of $\Upsilon$.  From this we see that the tensor network lives on this 
discrete version of the holographic slice.  Because isometric tensor 
networks obey a form of HRRT, one may mistakenly conclude that the cut 
through the network computes the area of the corresponding surface in 
the bulk along the holographic slice.  This is generally not the case; 
the maximin method tells us that this area provides only a lower bound 
on the entanglement.  The entanglement calculated in this way instead 
corresponds to the area of the HRRT surface anchored to the appropriate 
subregion of a renormalized leaf.  In other words, the tensor network 
should not be viewed as a discretization of the holographic slice, but 
rather as a set of boundary states dual to successively smaller domains 
of dependence.
\begin{figure}[t]
  \includegraphics[scale=0.25]{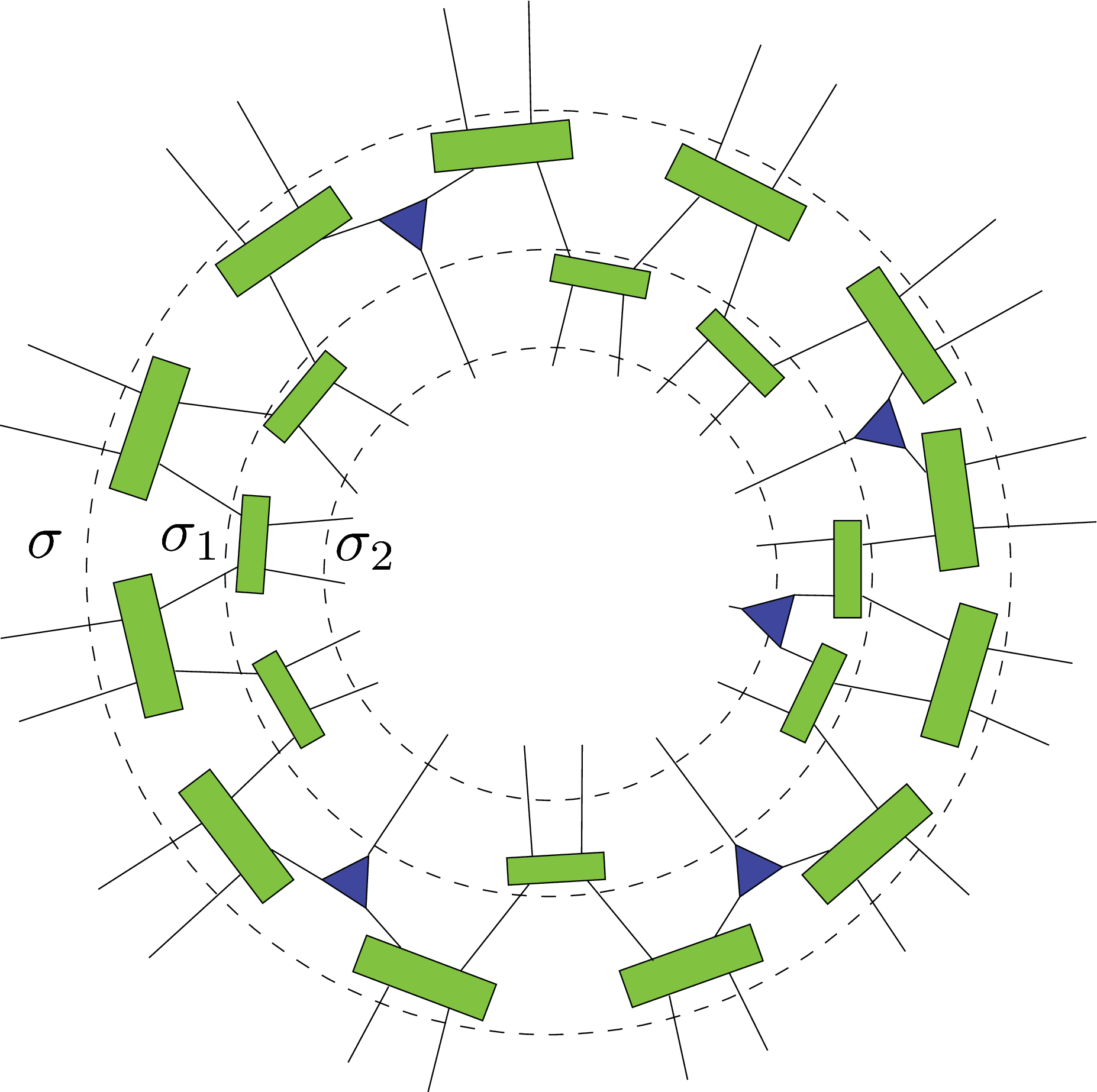}
\caption{A tensor network for a non-hyperbolic geometry.  The green 
 rectangles correspond to disentanglers while the blue triangles are 
 coarse-graining isometries.  Each internal leg of the tensor network 
 has the same bond dimension.  We are imagining that $\sigma$ corresponds 
 to a leaf of a holographic screen and each successive layer ($\sigma_1$ 
 and $\sigma_2$) is a finite size coarse-graining step of the holographic 
 slice.  Through this interpretation, the tensor network lives on the 
 holographic slice.  However, the entanglement entropy calculated via 
 the min-cut method in the network does not correspond to the distance 
 of the cut along the holographic slice in the bulk.  It corresponds to 
 the HRRT surface in the appropriate domain of dependence.  The locations 
 of $\sigma_1$ and $\sigma_2$ in the bulk are found by convolving the 
 HRRT surfaces for the small regions being disentangled and coarse-grained. 
 The holographic slice is a continuous version of this tensor network.}
\label{fig:TNetwork}
\end{figure}

In fact, this interpretation can be applied to any isometric tensor network, 
and we argue that this is the proper way to view tensor networks representing 
bulk spacetimes.  That is, given a state represented by an isometric tensor 
network, one can find a set of states by pushing the boundary state through 
the tensors one layer at a time such that no two layers have the same 
boundary legs.%
\footnote{This picture can perhaps be used to show that the dynamical 
 holographic entropy cone is contained within the holographic entropy 
 cone~\cite{Bao:2015bfa}.  By explicitly constructing the holographic 
 slice tensor network that encodes a state of a time dependent geometry, 
 we will have found a model that encodes the entropies in a way that 
 ensures containment within the holographic entropy cone.}
These successive states are then dual to bulk domains of dependence that 
are successively contained in each other, and whose boundaries lie on 
(a discrete version of) the holographic slice.

\subsection{Time Evolution and Gauge Fixing}
\label{subsec:gauge-fix}

The preferred holographic slice of Section~\ref{subsubsec:unique} provides 
us with a novel way to foliate spacetimes.  By applying the holographic 
slice procedure to each boundary time slice, one foliates the bulk 
spacetime with holographic slices.

In order for this to provide a good gauge fixing, the holographic slices 
generated from different boundary time slices must not intersect.  In 
spherically symmetric cases, these intersections do not occur.  From the 
spherical symmetry of the spacetime, the renormalized leaves must also 
be spherical.  Thus, if two holographic slices did intersect, they must 
intersect at a renormalized leaf.  However, the evolution of the slice 
is unique from this leaf, and hence these two slices do not intersect. 
Furthermore, the reverse flow is also unique and hence the slices must 
exactly coincide.  This prevents ambiguities in the gauge fixing of the 
bulk spacetime.

Outside of spherically symmetric cases, if $\text{sgn}(K)$, the sign of 
the extrinsic curvature of the slice, is constant over the slice then 
no slices will intersect.  By joining a slice with the past (future) 
portion of the holographic screen in the $K \ge 0$ ($\le 0$) case, one 
can create a barrier for extremal surfaces anchored in the interior 
of the barrier.  This implies that any slice constructed from a leaf 
cannot penetrate slices generated from leaves towards its future (past) 
if $K \ge 0$ ($\le 0$).

As explained in Sections~\ref{subsec:prop} and \ref{sec:example}, the 
foliation generated by the holographic slices will not probe behind 
late-time horizons.  Thus this foliation provides a gauge fixing of 
the region of spacetime exterior to black hole horizons.  In this 
region, the foliation provides a covariant map from boundary time 
slices to bulk time slices.

\section{Relationship to Renormalization}
\label{sec:renorm}

In this section, we discuss how our coarse-graining procedure is related 
to conventional renormalization, both in the context of standard quantum 
field theories and AdS/CFT.

\subsection{Analogy to Renormalization in Quantum Field Theories}
\label{subsec:analogy}

We first draw an analogy between pulling in the boundary along the 
holographic slice and standard renormalization in quantum field theories. 
In particular, we liken the limitations of fixed order perturbation 
theory with the existence of reconstructable shadows.  We begin by 
reviewing renormalization in quantum field theories phrased in a way 
to make the relationship clear.

Suppose one computes the amplitude of a process involving two widely 
separated mass scales $m$ and $E$ in fixed order perturbation theory. 
In terms of a renormalized coupling constant $g$, it is given generally 
in the form
\begin{equation}
  {\cal M} = \sum_{n=0}^\infty c_n 
    \biggl( \frac{g}{16\pi^2} \ln\frac{E}{m} \biggr)^n,
\label{eq:M-fixed}
\end{equation}
where $c_n$'s are of the same order.  This implies that even if 
$g/16\pi^2$ is small, this perturbation theory breaks down when 
$\ln(E/m) \sim 16\pi^2/g$.

There is, however, a way to resum these logarithms---the renormalization 
group.  Introducing the concept of running coupling 
constant, $g(\mu)$, defined at a sliding scale $\mu$, the amplitude 
of Eq.~(\ref{eq:M-fixed}) can be written as
\begin{equation}
  {\cal M} = \sum_{n=0}^\infty c_n 
    \biggl( \frac{g(\mu)}{16\pi^2} \ln\frac{E}{\mu} \biggr)^n.
\label{eq:M-RG}
\end{equation}
The process can now be calculated perturbatively as long as both $g(m)/16\pi^2$ 
and $g(E)/16\pi^2$ are small, where $g(m)$ and $g(E)$ are related by a 
continuous renormalization group evolution.  In general, the range of 
validity of this renormalization group improved perturbation theory is 
larger than that of fixed order perturbation theory.

This phenomenon is analogous to the existence of shadows in the holographic 
reconstruction.  If one tries to reconstruct the bulk in a ``single shot'' 
using HRRT surfaces anchored to the original leaf, then there can be 
regions in spacetime (shadows) that cannot be reconstructed.  This, however, 
is not a fundamental limitation of the perturbative reconstruction of 
the bulk.  As we have seen, we can reconstruct a portion of entanglement 
shadows by performing a reconstruction in multiple steps:\ first renormalizing 
the leaf and then using HRRT surfaces anchored to the renormalized leaf. 
By doing this renormalization with more steps, one can progressively probe 
deeper into shadow regions.  Going to the continuum limit (the holographic 
slice), we find that we can describe physics in shadows without difficulty.

Even with the renormalization group improvement, the perturbative description 
of physics stops working when $g(\mu)$ hits a Landau pole or approaches 
a strongly coupled fixed point.  This is analogous to the fact that the 
evolution of the holographic slice halts, $U(\lambda_1, \lambda_2) \propto 
\mathds{1}$ in Eq.~(\ref{eq:unitary-def}), once the renormalized leaf 
contracts to a point or approaches a horizon.  Incidentally, this picture 
is consonant with the idea that describing the interior of a black hole 
would require ``nonperturbative'' physics.%
\footnote{In particular, an interior description may require changing the 
 basis of multiple black hole microstates~\cite{Nomura:2012ex,Bao:2017who}, 
 each of which can be viewed as having different background geometries 
 corresponding to slightly different black hole masses~\cite{Nomura:2014woa,%
 Nomura:2016qum}.}

We stress that from the boundary point of view, renormalization of a 
leaf corresponds to coarse-graining of a state at a fixed time.  A natural 
question is if there is an effective theory relating coarse-grained states 
at different times.  We do not see a reason to doubt the existence of such 
a theory, at least for degrees of freedom sufficiently deep in the bulk. 
Since the coarse-graining depends on the state, however, the resulting 
description may well be applicable only within a given geometry, i.e.\ 
a selected semiclassical branch of the fundamental state in quantum gravity.

\subsection{Comparison to Holographic Renormalization in AdS/CFT}
\label{subsec:rel-AdSCFT}

How is the holographic slice related to holographic renormalization 
in asymptotically AdS spacetimes?  There has been extensive literature 
devoted to the latter subject.  Here we will highlight the essential 
difference between our renormalization procedure and the perspective 
of Susskind and Witten~\cite{Susskind:1998dq}.

If one kept $N^2$ (the number of field degrees of freedom) fixed per cutoff 
cell and applied the Susskind-Witten method of regularization deeper in 
the bulk, a local description of physics on the boundary would break down 
once $R \approx l_{\rm AdS}$.  Here, $R$ and $l_{\rm AdS}$ are the radius 
of the cutoff surface and the AdS length scale, respectively.  This is 
because the number of degrees of freedom within an $l_{\rm AdS}$-sized 
bulk region is of order $(l_{\rm AdS}/l_{\rm P})^{d-1}$, which is just 
$N^2$.  Here, $l_{\rm P}$ is the Planck length in the bulk.

However, holography extends to sub-AdS scales, and the extremal surfaces 
anchored to a cutoff at $l_{\rm AdS}$ satisfy the appropriate properties 
to be interpreted as entanglement entropies~\cite{Miyaji:2015yva,%
Sanches:2016sxy}.  As emphasized throughout the text, the connection 
between entanglement and geometric quantities seems to extend beyond 
AdS/CFT~\cite{Hayden:2016cfa,Harlow:2016vwg}.  Because of this, we expect 
that there should be some way to renormalize the boundary state in such 
a way to preserve the HRRT prescription at all scales.  This is, however, 
prohibited if we fix $N^2$ because already at an $l_{\rm AdS}$-sized 
region we lose the ability to talk about the entanglement of boundary 
subregions (as there is only one boundary cell).

Therefore, if one wants to preserve the ability to use the HRRT 
prescription, $N^2$ must change as the boundary is pulled in.  Simply 
requiring that $N^2 \ge 1$ per cell will allow renormalization down to 
$l_{\rm P}$.  This is easily seen by noticing that the number of cells 
for a boundary moved in to radius $R$ ($\leq l_{\rm AdS}$) is given 
by the total number of bulk degrees of freedom, $(R/l_{\rm P})^{d-1}$, 
divided by $N^2$.  This implies that when $R = l_{\rm P}$ the number 
of cells is of order unity, and the holographic description must break 
down.  In fact, the bulk description is expected to break down before 
this happens.  Suppose that the gauge coupling, $g$, of the boundary 
theory stays constant.  Then the requirement of a large 't~Hooft 
coupling, $N^2 \ge 1/g^4 = (l_{\rm s}/l_{\rm P})^{d-1}$, implies that 
the bulk spacetime picture is invalidated when $R \lesssim l_{\rm s}$. 
Here, $l_{\rm s}$ is the string length.  Assuming the existence of 
a renormalization scheme preserving the HRRT prescription implies that 
there exists a way to redistribute the original $N^2$ degrees of freedom 
spatially on a coarse-grained holographic space.

The construction of the holographic slice requires extremal surfaces 
to be anchored to renormalized leaves, so the renormalization procedure 
utilized must necessarily preserve the ability to use the HRRT prescription. 
The holographic slice, therefore, must employ the special renormalization 
scheme described above.

\section{Discussion}
\label{sec:disc}

The holographic slice is defined using HRRT surfaces, and hence is 
inherently background dependent.  This prohibits the use of the holographic 
slice as some way to analyze the coarse-grained behavior of complex 
quantum gravitational states with no clear bulk interpretation.  In 
particular, if a state is given by a superposition of many different 
semiclassical geometries,
\begin{equation}
  \ket{\Psi} = c_1 \ket{\psi_1} + c_2 \ket{\psi_2} + \cdots,
\label{eq:branch}
\end{equation}
then the holographic slice prescription can be applied to each branch 
of the wavefunction, $\ket{\psi_i}$, independently.  However, there is 
no well-defined slice for $\ket{\Psi}$.  This is the same limitation 
one would face when considering the entanglement wedge of similar states. 
Despite this, for superpositions of states within the code subspace, 
the analyses of Refs.~\cite{Nomura:2017npr,Almheiri:2016blp} tell us 
the holographic slice construction is well defined.

Regardless, the holographic slice sheds light on the nature of bulk 
emergence.  The construction of the slice harvests short range entanglement 
between small subregions, not in the form of entanglement entropy.  It 
is precisely this that allows the slice to flow into the bulk and through 
entanglement shadows.  This work emphasizes the idea that entanglement 
entropy as measured by von~Neumann entropy is not sufficient to 
characterize the existence of a semiclassical bulk viewed from the 
boundary.  Other measures of entanglement (negativity, entanglement 
of purification, etc.) may be more useful to analyze bulk emergence. 
This was explored extensively in Ref.~\cite{Nomura:2017fyh}.

The slice additionally provides a very natural interpretation for 
non-minimal extremal surfaces as the entanglement entropy for subregions 
of coarse-grained states.  Because the coarse-graining procedure 
mixes up the boundary degrees of freedom while removing the short 
range information, the interpretation of non-minimal extremal surfaces 
in terms of purely UV boundary terms will necessarily be very 
complicated~\cite{Balasubramanian:2014sra}.  However, once coarse-graining 
occurs these complicated quantities manifest with a simple interpretation. 
This is also what is seen in the entanglement of purification calculations.

By assuming that the holographic states all live within the same infinite 
dimensional Hilbert space, ${\cal H}_{\rm UV}$, we were able to discuss 
the mapping from a boundary state to a coarse-grained version of 
itself.  This is what gave rise to the $K(\lambda)$ operator in 
Section~\ref{subsec:rs-corresp}.  Alternatively, rather than use 
${\cal H}_{\rm UV}$ to discuss coarse-graining, one can use it to 
talk about time evolution in the boundary theory.  One of the major 
hurdles in formulating theories for holographic screens is the fact 
that the area of the screens are non-constant.  If one were to view 
this area as determining the size of the true Hilbert space the state 
lived in, then time evolution would require transitions between Hilbert 
spaces.  However, by viewing the leaf area as a measure of the size of 
the effective subspace that the state lives in, we are free from this 
complication.  In fact, modeling time evolution is similar to performing 
the reverse of the coarse-graining operation.  This introduces 
entanglement at shorter and shorter scales, which increases the 
effective subspace's size.  Of course, time evolution must account 
for other complex dynamics, but simply increasing the screen area 
is no difficulty.  This interpretation suggests that the area of 
holographic screens is a thermodynamic entropy measure, rather than 
a measure of the fundamental Hilbert space size.

Concluding, the holographic slice is a novel, covariantly defined 
geometric object.  It encodes the bulk regions dual to successively 
coarse-grained states and we propose that the flow along the slice is 
governed by distillable correlations at the shortest scales.  This may 
be related to the entanglement of purification of small regions or the 
modular evolution of such regions.  Investigation of the explicit boundary 
flow along the slice seems to be the most promising avenue of future 
work.  It may also be fruitful to study the mean curvature vector flow 
of \mbox{codimension-2} convex surfaces in Lorentzian spacetimes, as 
characterizing solutions to this flow may provide insights into the 
coarse-graining operation.

\acknowledgments

We thank Chris Akers, Ning Bao, Raphael Bousso, Venkatesa Chandrasekaran, 
Grant Remmen, Fabio Sanches, and Arvin Shahbazi-Moghaddam for discussions. 
This work was supported in part by the National Science Foundation under grants 
PHY-1521446, by MEXT KAKENHI Grant Number 15H05895, and by the Department 
of Energy, Office of Science, Office of High Energy Physics, under 
contract No.\ DE-AC02-05CH11231.

\appendix

\section{Intersection of Domains of Dependence}
\label{app:domains}

\begin{lemma}
Let $\Sigma$ be a closed, achronal set and $D(\Sigma)$ be the domain of 
dependence of $\Sigma$.  Let $p$ and $q$ be points in $D(\Sigma)$, and 
$\lambda$ a causal curve such that $\lambda(0)=p$ and $\lambda(1)=q$ 
where $p$ lies to the past of $q$.  Then, all points $r=\lambda(t)$ for 
$t\in\left[0,1\right]$ are contained in $D(\Sigma)$.
\label{lem:1}
\end{lemma}
\begin{proof}
Suppose such a point $r$ does not belong to $D(\Sigma)$.  Then, there must 
exist an inextendible causal curve $\lambda'$ that passes through $r$ and 
does not intersect $\Sigma$.  Without loss of generality, assume that $r$ 
lies to the past of $\Sigma$.  Consider a causal curve composed of $\lambda$ 
to the past of $r$ and $\lambda'$ to the future of $r$.  This would then 
be an inextendible causal curve passing through $p$ but not intersecting 
$\Sigma$, implying that $p$ does not belong to $D(\Sigma)$, thus 
contradicting the assumption.
\end{proof}
\begin{lemma}
Let $R$ be a closed set such that every causal curve connecting two points 
in $R$ lies entirely in $R$.  Let $\Sigma$ be the future boundary of $R$ 
defined by points $p \in R$ such that $\exists$ a timelike curve $\lambda$ 
passing through $p$ that does not intersect $R$ anywhere in the future. 
Then,
\begin{enumerate}
  \item $\Sigma$ is an achronal set.
  \item $R \subseteq D(\Sigma)$.
\end{enumerate}
\label{lem:2}
\end{lemma}
\begin{proof}
We first show that $\Sigma$ is an achronal set.  Suppose there exist two 
points $p$ and $q$ in $\Sigma$ that were timelike related.  Without loss 
of generality, assume that $p$ lies to the past of $q$.  Consider an open 
neighborhood of $p$ denoted by $U(p)$.  Consider a point $r$ such that 
$r \in \left\{ I_+(p) \cap U(p) \right\} \backslash R$.  By continuity, 
$\exists$ a timelike curve $\lambda$ connecting $r$ to $q$.  $\lambda$ can 
then be extended to pass through $p$ in the past.  Thus, we have found a 
causal curve that connects points $p$ and $q$, both of which belong to $R$, 
and passes through $r \notin R$.  This contradicts the assumption, and hence, 
$\Sigma$ must be an achronal set.

Now, we can show that $R \subseteq D(\Sigma)$.  Consider a point $p$ such 
that $p \in R\, \backslash \Sigma$. Then, $I_+(p)$ must intersect $\Sigma$. 
To show this, suppose it were not true and consider a future causal curve 
$\lambda$ from $p$ which does not intersect $\Sigma$.  The boundary point 
of $\lambda \cap R$ then also has a timelike curve through it which does not 
intersect $R$ anywhere in the future, and thus should be included in the set 
$\Sigma$.  Therefore, $I_+(p)$ either intersects $\Sigma$ everywhere in the 
interior of $\Sigma$ or intersects some portion of the boundary of $\Sigma$. 
In the first case, all inextendible causal curves through $p$ necessarily 
pass through $\Sigma$, and hence $p \in D(\Sigma)$.  In the second case, 
extend $\Sigma$ in a spacelike manner to an open neighborhood around $\Sigma$ 
where a point $q$ outside $\Sigma$ is timelike related to $p$.  Causal curves 
from $p$ to $q$ would not intersect $\Sigma$ since $q$ is spacelike related 
to all points on $\Sigma$.  Consider the intersection of this curve with $R$. 
It must have a boundary point which does not belong to $\Sigma$.  This point 
would then have inextendible timelike curves through it that do not intersect 
$R$ in the future.  This contradicts the assumption that this point was not 
in $\Sigma$. This implies that the second case is impossible.  Hence, we have 
proved that $R \subseteq D(\Sigma)$.
\end{proof}
\begin{theorem}
Consider two \mbox{codimension-1} spacelike subregions $\Sigma_1$ 
and $\Sigma_2$ that are compact.  Let their domains of dependence be 
$D(\Sigma_1) = D_1$ and $D(\Sigma_2) = D_2$.  Then $D = D_1 \cap D_2$ 
is the domain of dependence of the future boundary of $D$ denoted 
by $\Sigma$.
\end{theorem}
\begin{proof}
Consider any two points $p$ and $q$ that belong to $D$.  Both $p$ and $q$ 
belong to $D_1$ and $D_2$.  Using Lemma~\ref{lem:1}, we can conclude that 
all points on a causal curve joining $p$ and $q$ belong to both $D_1$ and 
$D_2$.  Hence, any such point also belongs to $D$.  Thus, $D$ satisfies the 
condition required for $R$ above in Lemma~\ref{lem:2}. Using Lemma~\ref{lem:2} 
then tells us that $D \subseteq D(\Sigma)$.

Since $\Sigma$ is defined to be the future boundary of $D$, $\Sigma$ itself 
is necessarily contained in $D$.  Now consider any $p \in D(\Sigma)$.  Any 
causal curve $\lambda$ passing through $p$ intersects $\Sigma$ by definition. 
However, since $\Sigma \subseteq D$, all inextendible causal curves through 
$\Sigma$ necessarily intersect both $\Sigma_1$ and $\Sigma_2$.  Thus, all 
inextendible causal curves through $p$ also pass through both $\Sigma_1$ 
and $\Sigma_2$.  This implies $D(\Sigma) \subseteq D$.

Combining the above two results, we have shown that $D = D(\Sigma)$.  Namely, 
the intersection of two domains of dependence is also a domain of dependence.
\end{proof}

\section{Uniqueness of the Holographic Slice}
\label{app:uniqueness}

Consider a \mbox{codimension-2}, closed, achronal surface $\sigma$ in an 
arbitrary $(d+1)$-dimensional spacetime $M$.  Suppose $\sigma$ is a convex 
boundary.  We assume that both $M$ and $\sigma$ are sufficiently smooth 
so that variations in the spacetime metric $g_{\mu\nu}$ and induced metric 
on $\sigma$, denoted by $h_{ij}$, occur on characteristic length scales 
$L$ and $L_\sigma$, respectively.

\begin{theorem}
Consider subregion $R$ of characteristic length $\delta \ll L, L_\sigma$ 
on $\sigma$.  To leading order, the extremal surface anchored to 
$\partial R$ lives on the hypersurface generated by the vector 
$s = \theta_t t - \theta_z z$ normal to $\sigma$.  Here, $t$ and $z$ 
are orthonormal timelike and spacelike vectors perpendicular to $\sigma$, 
and $\theta_t = h^{ij} K^t_{ij}$ and $\theta_z = h^{ij} K^z_{ij}$ where 
$K^t_{ij}$ and $K^z_{ij}$ are the extrinsic curvature tensors of $\sigma$ 
for $t$ and $z$, respectively.  This property is independent of the 
shape of $R$.
\label{thm:hypersurface}
\end{theorem}
\begin{proof}
Start from a point $p \in R$ and set up Riemann normal coordinates in 
the local neighborhood of $p$.
\begin{equation}
  g_{\mu\nu}(x) = \eta_{\mu\nu} 
    - \frac{1}{3} R_{\mu\rho\nu\sigma} x^\rho x^\sigma + O(x^3).
\label{eq:Riemann-normal}
\end{equation}
In these coordinates, we are considering a patch of size $\delta$ around 
the origin $p$ with $R_{\mu\rho\nu\sigma} \sim O(1/L^2)$.  Equivalently, 
we could consider a conformally rescaled metric
\begin{align}
  x^{\mu} &= \epsilon y^{\mu},
\label{eq:rescale-1}\\
  ds^2 &= \epsilon^2 g_{\mu\nu}(\epsilon y)\, dy^\mu dy^\nu,
\label{eq:rescale-2}\\
  d\tilde{s}^2 &= g_{\mu\nu}(\epsilon y)\, dy^\mu dy^\nu 
\nonumber\\
  &= \tilde{g}_{\mu\nu}(y)\, dy^\mu dy^\nu,
\label{eq:rescale-3}
\end{align}
where $\epsilon = \delta/L \ll 1$.

In this alternate way of viewing the problem, we have a patch of size $L$ 
with the metric varying on a larger length scale $L/\epsilon$.  In these 
coordinates, each derivative of the conformal metric brings out an extra 
power of $\epsilon$; for example,
\begin{equation}
  \frac{\partial^2}{\partial y^\rho \partial y^\sigma} \tilde{g}_{\mu\nu} 
  = \epsilon^2 \frac{\partial^2}{\partial x^\rho \partial x^\sigma} g_{\mu\nu} 
  \sim \frac{\epsilon^2}{L^2}.
\label{eq:metric-deriv}
\end{equation}
The connection coefficients $\Gamma^\mu_{\rho\sigma}$ vanish at $p$ due 
to our choice of Riemann normal coordinates.  This implies that for points 
in the neighborhood of $p$, we can Taylor expand to find
\begin{align}
  \Gamma^\mu_{\rho\sigma} &\sim \frac{\epsilon^2}{L},
\label{eq:connection}\\
  R_{\mu\rho\nu\sigma} &\sim \frac{\epsilon^2}{L^2}.
\label{eq:Ricci_tensor}
\end{align}
Note that these quantities are obtained using the rescaled metric 
$\tilde{g}_{\mu\nu}$ in the $y^\mu$ coordinates.

Since there is still a remaining $SO\left(d,1\right)$ symmetry that 
preserves the Riemann normal coordinate form of the metric, we can 
use these local Lorentz boosts and rotations to set $t$ and $z$ as 
the coordinates in the normal direction to $\sigma$ at $p$ while $y^i$ 
parameterize the tangential directions.  This is a convenient choice 
to solve the extremal surface equation in a perturbation series order 
by order.  The extremal surface equation is given by~\cite{Hubeny:2007xt}
\begin{equation}
  \tilde{g}^{\rho\sigma} \Bigl( \partial_\rho \partial_\sigma Y^\mu 
    + \Gamma^\mu_{\lambda\eta} \partial_\rho Y^\lambda \partial_\sigma Y^\eta 
    - \Gamma^\lambda_{\rho\sigma} \partial_\lambda Y^\mu \Bigr) = 0.
\label{eq:extremal-eq}
\end{equation}
This is a set of $d+1$ equations for the embedding of the extremal surface 
$Y^\mu$, which are functions of $d-1$ independent coordinates.  The 
equations in the tangential directions are trivially satisfied by taking 
the $d-1$ parameters to be $y^i$.  This leaves only two equations in 
the normal directions to be solved.

From the discussion above, when restricted to the local patch of size 
$L$, we have
\begin{align}
  \tilde{g}^{\mu\nu} &= \eta^{\mu\nu} + O(\epsilon^2),
\label{eq:local-q}\\
  \Gamma^\mu_{\rho\sigma} &= O\bigg(\frac{\epsilon^2}{L}\biggr).
\label{eq:local-connec}
\end{align}
Assuming the extremal surface is smooth, derivatives of $Y^\mu$ typically 
bring down a power of $L_\sigma$.  Thus,
\begin{align}
  \partial Y   &\sim O(1),
\label{eq:deriv-1}\\
  \partial^2 Y &\sim O\biggl(\frac{1}{L_\sigma}\biggr).
\label{eq:deriv-2}
\end{align}
Using this, at the leading order in $\epsilon$ and $\epsilon_\sigma = 
\delta/L_\sigma$, the extremal surface equations simply become
\begin{equation}
  \delta^{ij} \partial_i \partial_j Y^\mu = 0,
\label{eq:ESE-1}
\end{equation}
where $\mu$ takes the $t$ and $z$ directions.  We write these as
\begin{equation}
  \nabla^2 t_{\rm E} = \nabla^2 z_{\rm E} = 0,
\label{eq:ESE-2}
\end{equation}
where $t_{\rm E}$ and $z_{\rm E}$ are functions of $y^i$.

Let $K^t_{ij}$, $K^z_{ij}$ denote the extrinsic curvature tensors for the 
$t$ and $z$ normals, respectively.  Following the above scaling arguments, 
$K^t_{ij},\, K^z_{ij} \sim \epsilon/L_\sigma$.  Here, we have assumed 
that $L_\sigma \lesssim L$, although this is not essential for the final 
result.  Because $t$ and $z$ are normal to the leaf, the equations for 
the leaf, described by $t_{\rm L}(y^i)$ and $z_{\rm L}(y^i)$, can be 
Taylor expanded in the neighborhood $R$ as
\begin{align}
  t_{\rm L}(y^i) &= -\frac{1}{2} K^t_{ij} y^i y^j 
    + O\biggl( \frac{\epsilon^2 y^3}{L_\sigma^2} \biggr),
\label{eq:leaf-R-1}\\
  z_{\rm L}(y^i) &= \frac{1}{2} K^z_{ij} y^i y^j 
    + O\biggl( \frac{\epsilon^2 y^3}{L_\sigma^2} \biggr),
\label{eq:leaf-R-2}
\end{align}
where the negative sign in the first line is due to the timelike signature 
of the $t$ normal.  The boundary conditions for the extremal surface 
equation are
\begin{align}
  t_{\rm E}(\partial R) &= t_{\rm L}(\partial R),
\label{eq:ext-bc-1}\\
  z_{\rm E}(\partial R) &= z_{\rm L}(\partial R).
\label{eq:ext-bc-2}
\end{align}
Now, consider $\delta t = t_{\rm E} - t_{\rm L}$ and $\delta z = z_{\rm E} 
- z_{\rm L}$.  The extremal surface equations are then given by
\begin{align}
  \nabla^2\, \delta t &= -\nabla^2 t_{\rm L} 
  = \theta_t \left\{ 1 + O(\epsilon_\sigma) \right\},
\label{eq:delta_t}\\
  \nabla^2\, \delta z &= -\nabla^2 z_{\rm L} 
  = -\theta_z \left\{ 1 + O(\epsilon_\sigma) \right\},
\label{eq:delta_z}
\end{align}
where $\theta_t = h^{ij} K^t_{ij}$ and $\theta_z = h^{ij} K^z_{ij}$. 
Note that $h_{ij} = \eta_{ij}$ at this order.  The boundary conditions 
are given by
\begin{equation}
  \delta t\left(\partial R\right) = \delta z\left(\partial R\right) = 0.
\label{eq:delta-bc}
\end{equation}

It is now clear that at leading order $\delta t/\theta_t$ and 
$-\delta z/\theta_z$ satisfy the same equation with the same boundary 
conditions.  Thus,
\begin{equation}
  \frac{\delta t}{\delta z} = -\frac{\theta_t}{\theta_z} + O(\epsilon_\sigma),
\label{eq:on-plane}
\end{equation}
for all points on the extremal surface.  Rewritten, the extremal surface 
lives on the hypersurface generated by $s = \theta_t t - \theta_z z$, 
orthogonal to $\sigma$.

This result is independent of the explicit shape of subregion $R$.
\end{proof}

Theorem~\ref{thm:hypersurface} essentially brings us to the uniqueness 
of the holographic slice.  The new surface, $\sigma'$, is generated by 
a convolution of the ``deepest'' points on each $\gamma(R)$.  Considering 
balanced shapes such that the ``deepest'' point corresponds to $y^i = 0$, 
$\delta t/\delta z$ has the interpretation of the slope of the evolution 
vector $s$ from $p$ which takes it to the new leaf $\sigma'$.  Slight 
imbalances in the shape would only affect the slope at subleading order 
in $\epsilon, \epsilon_\sigma$ and thus, the slope of $s$ is determined 
in a shape independent manner in the limit $\epsilon, \epsilon_\sigma 
\rightarrow 0$.  In order to move to $\sigma'$, we must also specify 
the distance, $\delta\lambda(p)$, by which we move along $s$ at each step. 
If the size of $C(p)$ is homogeneous across $\sigma$, then $\delta\lambda(p)$ 
is independent of $p$ to leading order.  Thus, the new leaf $\sigma'$ 
obtained at each stage is unique up to small error terms.  Following 
a similar procedure at each stage, e.g.\ by choosing random uncorrelated 
shapes of size $\delta'$ (found by mapping length $\delta$ to $\sigma'$ 
by $s$) for subregions $C'(p)$ at each point $p$, ensures that the error 
terms do not add up coherently.  This implies that the holographic slice 
is obtained by following the integral curves of the evolution vector $s$ 
starting from each point $p \in \sigma$, and hence is unique.
\begin{corollary}
Construct a holographic slice such that $C^i(p)$ is homogeneous and 
uncorrelated with $C^j(p)$, $j \neq i$.  Let the sizes of $C^i(p)$ be 
determined by mapping the characteristic length, $\delta$, of $C(p)$ on 
$\sigma$ to $\sigma^i$ by $s$.  The continuum version (sending $\delta 
\rightarrow 0$) of all such slices are identical.
\end{corollary}

As an aside, there are certain interesting features that this analysis 
highlights.  Consider a generic leaf of a holographic screen $\sigma$ 
and the future directed orthogonal null vectors $k$ and $l$ normalized 
as $k \cdot l = -2$.  The $t$ and $z$ vectors are then given by
\begin{equation}
  t = \frac{1}{2}(k + l),
\qquad
  z = \frac{1}{2}(k - l).
\label{eq:t-z-vec}
\end{equation}
From the linearity of extrinsic curvature, this leads to
\begin{equation}
  \theta_t = \frac{1}{2} (\theta_k + \theta_l),
\qquad
  \theta_z = \frac{1}{2} (\theta_k - \theta_l).
\label{eq:theta_t-z}
\end{equation}
The evolution vector $s$ and its associated expansion $\theta_s$ are 
given by
\begin{equation}
  s = \theta_t t - \theta_z z 
  = \frac{1}{2} (\theta_k l + \theta_l k),
\label{eq:s-vec-2}
\end{equation}
and
\begin{equation}
  \theta_s = \theta_t^2 - \theta_z^2 = \theta_k \theta_l \leq 0,
\label{eq:theta_s-neg}
\end{equation}
respectively.

At the holographic screen, $\theta_k = 0$.  This leads to
\begin{equation}
  s \propto k,
\qquad
  \theta_s = 0.
\label{eq:s-theta_s-leaf}
\end{equation}
Namely, the initial evolution of the holographic slice from a 
non-renormalized leaf occurs in the $k$ direction with a non-expanding 
or contracting leaf area.

\section{Convexity of Renormalized Leaves}
\label{app:convexity}

\begin{definition*}
On a spacelike slice $\Sigma$, a compact set $S$ is defined to 
be convex if and only if $K_{\Sigma}(\partial S) \leq 0$, where 
$K_{\Sigma}(\partial S)$ is the trace of the extrinsic curvature of 
$\partial S$ embedded in $\Sigma$ for the normal pointing inward.
\label{lem:3}
\end{definition*}
\begin{definition*}
In a spacetime $M$, a \mbox{codimension-2} compact surface $\sigma$ is 
called a convex boundary if on every \mbox{codimension-1} spacelike slice 
$\Sigma$ such that $\sigma \subset \Sigma$, the closure of the interior 
of $\sigma$ is a convex set.
\end{definition*}
\begin{lemma}
Let us consider a sufficiently small codimension-2 region $A \subset 
D(\sigma)$. There is then an extremal surface $\gamma(A)$ anchored to 
$A$ which stays within $D(\sigma)$. A codimension-2 surface $\sigma$ 
is a convex boundary if and only if $\gamma(A)$ does not leave $D(\sigma)$ 
under any continuous deformation of $A$ within $D(\sigma)$; see 
Ref.~\cite{Sanches:2016sxy} for details.
\label{lem:3}
\end{lemma}
\begin{lemma}
$\sigma$ is a convex boundary if and only if the null expansions in the 
inward direction, i.e.\ $\theta_{k}$ and $\theta_{-l}$, are both non-positive.
\end{lemma}
\begin{proof}
An inward normal $n$ on a spacelike slice $\Sigma$ is given by a linear 
superposition of $k$ and $l$, i.e.\ $n = \alpha k - \beta l$ with some 
$\alpha,\beta \geq 0$.  If $\theta_k \leq 0$ and $\theta_l \geq 0$, then 
$K_{\Sigma}(\sigma) = \theta_n = \alpha\theta_k - \beta\theta_l \leq 0$ 
for all choices of $\alpha, \beta \geq 0$.  Thus, by the above definition, 
$\sigma$ would be a convex boundary.  For the converse, suppose 
$\theta_k > 0$.  One can then choose $\Sigma$ such that 
$K_{\Sigma}(\sigma) > 0$ by taking $\beta \ll \alpha$.  Thus, 
$\sigma$ would not be convex on $\Sigma$, and hence $\sigma$ would 
not be a convex boundary.  The same argument applies if $\theta_l < 0$.
\end{proof}

\begin{fact}
A leaf of a holographic screen is a convex boundary.  The boundary of any 
entanglement wedge is also a convex boundary.
\end{fact}
\begin{theorem}
The intersection of the interior domains of dependence of two convex 
boundaries $\sigma_1$ and $\sigma_2$, represented by $D_1$ and $D_2$, 
is the interior domain of dependence of a convex boundary $\sigma'$.
\end{theorem}
\begin{proof}
As shown in Appendix~\ref{app:domains}, $D' = D_1 \cap D_2$ is the interior 
domain of dependence of some $\sigma'$.  We thus only need to show that 
$\sigma'$ is convex.  In order to do this, we can consider a sufficiently 
small \mbox{codimension-2} region $A \subset D'$. There is then an 
extremal surface $\gamma(A)$ anchored to $A$ which stays within $D'$. 
Let us now deform $A$ continuously such that it stays within $D'$. By 
the convexity of $\sigma_1$, the surface $\gamma(A)$ does not go outside 
$D_1$. Similarly, it does not go outside $D_2$ either, by the convexity 
of $\sigma_2$. This implies that $\gamma(A)$ does not go outside $D'$, 
which in turn implies that $\sigma'$ is convex by Lemma~\ref{lem:3}.
\end{proof}
\begin{corollary}
In a black hole spacetime or the case of a spacelike screen in an FRW 
spacetime, the coarse-graining procedure moves away from the singularity 
in the direction where the expansions $\theta_k$ and $\theta_l$ have 
opposite signs.  At each step of coarse-graining, $\theta_k$ and $\theta_l$ 
in general have opposite signs.
\end{corollary}

\end{document}